\documentclass[11pt,a4paper]{article}

\usepackage[margin=1in]{geometry}

\usepackage[utf8]{inputenc}
\usepackage[english]{babel}
\usepackage[dvipsnames,usenames]{xcolor}
\usepackage[colorlinks=true,pdfpagemode=UseNone,linkcolor=blue,citecolor=blue,pdfstartview=FitH]{hyperref}

\usepackage{subcaption}

\usepackage{mathtools}
\usepackage{amsthm,thmtools} 
\usepackage{amsmath,amssymb}
\usepackage{amsfonts}
\usepackage{bm}
\usepackage{booktabs}

\usepackage[nameinlink]{cleveref}

\usepackage{tikz}
\usepackage{pgfplots}
\pgfplotsset{compat=newest}
\usetikzlibrary{shapes, shapes.symbols, backgrounds, matrix, calc, arrows, math, arrows.meta, decorations.pathreplacing, decorations.markings, patterns, calligraphy}

\usepackage{booktabs}

\usepackage{nicefrac}

\usepackage{natbib}

\usepackage{enumitem}
\usepackage[ruled]{algorithm2e} 

\SetAlFnt{\small}
\SetAlCapFnt{\small}
\SetAlCapNameFnt{\small}
\SetAlCapHSkip{0pt}
\IncMargin{-\parindent}

\providecommand{\email}[1]{\href{mailto:#1}{\nolinkurl{#1}\xspace}}

\usepackage{acronym}

\newtheorem{theorem}{Theorem}

\newtheorem{proposition}[theorem]{Proposition}

\newtheorem{definition}[theorem]{Definition}

\newtheorem{corollary}[theorem]{Corollary}

\newtheorem{remark}{Remark}
\title{Check, Please: Verifiably Fair Clustering}

\author{Yu He, Jeremy Vollen, Edith Elkind\\
    Northwestern University  \\
    \texttt{yuhe2030@u.northwestern.edu},
    \texttt{\{vollen ,edith.elkind\}@northwestern.edu}
}
\date{}

\begin{document}

\maketitle

\begin{abstract}
Popular centroid-based clustering methods are typically optimized for global objectives, and may fail to adequately represent large groups of datapoints. Thus, one needs notions of proportionality that are suited for metric settings. Ideally, such notions should 
admit polynomial-time algorithms for
(a) finding proportional outcomes, and (b) checking if a given outcome is proportional; the latter property enables us to evaluate the outputs of traditional algorithms that do not offer proportionality guarantees (such as e.g., $k$-means). 
A promising approach is to import proportionality notions from the setting of multiwinner voting with approval ballots.
In particular, mPJR---the metric version of the well-known Proportional Justified Representation (PJR) axiom---satisfies (a), but up until now it was not known whether it satisfies (b).
In this work, we study the computational complexity of auditing proportional representation in clustering. 
In the approval setting, PJR is known to be coNP-complete to verify; however, it admits a strengthening, known as PJR+, which satisfies (a) and (b).
We show that these results translate to the metric setting: we prove that mPJR is coNP-complete to verify, define a metric analog of PJR+, which we call mPJR+, and argue that mPJR+ satisfies (a) and (b). However, the algorithm for auditing (m)PJR+ relies on repeated submodular minimization, rendering it impractical at scale, and we show that a natural combinatorial approach to verifying mPJR+ is infeasible. 
As a partial remedy, we propose an mPJR+ verification algorithm whose running time is exponential in the number of centers $k$, but quasilinear in the number of datapoints; it is practical if $k$ is small, but does not scale well with $k$.
Motivated by these hardness results, we introduce Default Coalitions mPJR+ (DC-mPJR+): a new proportionality concept that offers representation guarantees to a restricted set of coalitions around unselected centers, and, as a result, admits an $O(mn \log n + mnk)$ verification algorithm. We establish that DC-mPJR+ outcomes can be computed efficiently and that DC-mPJR+ remains a meaningful proxy for global fairness: any solution satisfying $\gamma$-DC-mPJR+ also satisfies $(\gamma + 2)$-mPJR+. 
\end{abstract}

\section{Introduction}
\label{sec:intro}

Clustering is one of the most fundamental tasks in computer science, with applications spanning unsupervised learning, data compression, information retrieval, image analysis, and computational biology. 
A particularly prominent clustering paradigm is {\em centroid-based clustering}, where one selects a set of representative centers and assigns each datapoint to a center. 
The quality of an outcome is typically measured by a global objective that aggregates the distances from each datapoint to its nearest cluster center using a suitable operator, such as sum ($k$-means), sum of squares ($k$-median), or max ($k$-center)
\citep{jain1999data,XuWunsch2005SurveyClustering,AggarwalReddy2014DataClustering}.

However, a growing body of work argues that, in many applications, centroid-based clustering should be evaluated not only by aggregate fit, but also by the way representation is distributed across different parts of the dataset. Specifically, it may be desirable to ensure that sufficiently large and similar groups of datapoints are allocated a fair share of cluster centers. 
This perspective connects clustering to broader questions of fairness, and in particular to notions of proportionality, which require that each group receives representation commensurate with their size. 

The study of proportionality is a burgeoning subarea of (computational) social choice. It dates back to \citet{dummett1984voting}, who argued that a democratically elected committee should accurately represent the ballots submitted by its constituents. Following this intuition, representation is defined not in terms of pre-specified groups of voters, but rather in terms of coalitions that share similar preferences over the available candidates. For multiwinner voting with approval ballots, this perspective inspired the family of justified representation axioms \citep{aziz2017justified,sanchez2026proportional,PSP21a,brill2023robust,KalayciL025}; for voting with ranked ballots it led to the concept of Proportionality for Solid Coalitions \citep{dummett1984voting,AzLe20a}.

Given this rich literature, a natural approach to defining proportionality in the context of clustering is by means of importing notions of proportional representation from social choice \citep{chen2019proportionally, micha2020proportionally, aziz2024proportionally}.
The most prominent concept that was derived in this manner is \emph{Proportionally Representative Fairness} (PRF) \citep{aziz2024proportionally}, which is closely connected to the notion of Proportional Justified Representation (PJR) \citep{sanchez2026proportional} from multiwinner voting: 
\citet{kellerhals2024proportional} formalized mPJR as the metric analog of PJR and showed its equivalence to PRF. 

Recent algorithms specifically designed to enforce these representation guarantees---such as the Greedy Capture of \citet{chen2019proportionally} and the Spatial Expanding Approval Rule (SEAR) of \citet{aziz2024proportionally}---demonstrate that strong fairness and representation properties, including PRF, can be achieved in an algorithmically efficient way. 
However, the classical objective-driven approaches, especially $k$-means and related center-based methods,  
remain dominant in practice because of their simplicity, scalability, and long record of empirical usefulness \citep{jain1999data,berkhin2006survey,ghadiri2021socially}.
This is because in many real-world use cases, proportional representation is a secondary design goal rather than the primary optimization criterion. 
Accordingly, practitioners may be unwilling to abandon classical methods entirely, but may still wish to understand whether the solutions they already compute provide adequate representation in practice. 
Thus, we need to develop tools that can audit the outputs of standard centroid-clustering algorithms, i.e., determine whether a given set of cluster centers satisfies our proportionality criteria.

Motivated by these concerns, we study the problem of verifying proportional representation for centroid clustering outcomes, with the goal of identifying a concept that simultaneously (1) captures a meaningful proportional representation requirement, (2) is satisfied by a polynomial-time algorithm, and (3) admits a verification algorithm fast enough to be practical on large-scale clustering datasets.

\subsection{Our Contributions}
We first prove that PRF --- the only extant property capturing proportional representation in centroid-based clustering that is guaranteed to be satisfiable on every instance --- is coNP-complete to verify. 
To do so, we give a general reduction from approval-based committee voting to centroid-based clustering. Our result then follows from coNP-completeness of verifying PJR in the approval setting \citep{aziz2018complexity}.

Next, we turn to defining a proportional representation property that is polynomial-time verifiable.
For the approval setting, \citet{brill2023robust} propose PJR+: a strengthening of PJR that is both polynomial-time computable and polynomial-time verifiable.
We formulate a metric analog of PJR+, which we call mPJR+, and show that it is indeed polynomial-time verifiable, and moreover, that it is satisfied by SEAR.

However, the verification algorithm of \citet{brill2023robust} (as well as our adaptation of this algorithm to the metric setting) relies on solving a sequence of submodular minimization problems, which renders it unrealistic at scale.
Furthermore, we show that a natural approach toward a more efficient verification algorithm for mPJR+, namely, conditioning on the  representation level $\ell$ that a group may demand based on its size, 
requires solving a series of coNP-complete problems. 

As a partial remedy, we propose an (m)PJR+ verification algorithm whose running time is exponential in the number 
of centers $k$, 
but quasilinear in $n$; this approach 
is useful if the target number of clusters is small, but becomes impractical if $k$ is large.

Thus, we search for a new property that can be verified tractably for large datasets.
We define \emph{ Default Coalitions mPJR+} (DC-mPJR+), which relaxes mPJR+ by restricting guarantees to coalitions of a particular structure --- for each unselected center and representation level $\ell$, we only give guarantees to the tightest group of agents around that center large enough to justify $\ell$ representatives. 
We establish that DC-mPJR+ is efficiently verifiable, by presenting an $O(mn\log n + mnk)$ verification algorithm. 
We also demonstrate that DC-mPJR+ meaningfully captures proportional representation: we argue that it is equivalent to mPJR+ up to a constant multiplicative factor, with the equivalence extending to approximate versions of the two properties. Together, our results identify a practical and theoretically grounded path for auditing proportional representation in clustering.

\subsection{Related Work}
\label{sec:related_work}

\paragraph{Proportional representation in social choice.}
Proportional representation is a central principle in collective decision-making, which posits that large groups of voters with sufficiently similar preferences should receive representation commensurate with their size. 
In approval-based multiwinner voting, this idea is captured by a hierarchy of \emph{justified representation} properties, such as JR, PJR, EJR, PJR+, EJR+, FPJR, FJR, and the core (see \citet{lackner2023multi} for a survey of this area). 
Beyond approval-based preferences, similar notions have been defined in settings with additive utilities \citep{PSP21a}, rankings \citep{dummett1984voting}, and weak orders \citep{AzLe20a,brill2023robust}.
However, even in the approval-based setting, verification of PJR and EJR is coNP-complete \citep{aziz2018complexity}. This motivated \citet{brill2023robust} to introduce stronger and polynomial-time verifiable properties called PJR+ and EJR+.
Our work imports this verification perspective into metric clustering.

\paragraph{Proportionality and fairness in clustering.}
In the context of clustering, fairness may be interpreted as achieving balanced representation with respect to exogenous attributes \citep[see, e.g.,][]{chierichetti2017fair,ahmadian2019clustering,abbasi2021fair}, or as allocating cluster centers to groups of datapoints in a way that is commensurate with the group size and diameter
\citep[see, e.g.,][]{chen2019proportionally,jung2020service,li2021approximate,micha2020proportionally,aziz2024proportionally,kellerhals2024proportional}. 
Our work belongs to the latter strain of research.

\citet{chen2019proportionally} introduced the idea of proportional fairness for clustering and put forward the Greedy Capture algorithm for computing approximately proportional solutions. 
They also consider the verification problem and give efficient algorithms for it.
However, it has been argued that proportional fairness --- along with other similar notions \citep{jung2020service,li2021approximate} --- is not expressive enough to address concerns of proportional representation \citep{aziz2024proportionally}.
This motivated \citet{aziz2024proportionally} to introduce PRF: a property that captures proportional representation, is guaranteed to be satisfiable, and is computable in polynomial time.
Other works have contributed definitions that are similar in spirit, such as $(\alpha,q)$-core \citep{ebadian2025boosting} and $(\alpha,\gamma)$-proportional representation \citep{kalayci2024proportional}.
In this work, we focus on PRF because of its existence guarantees; however, considering these other properties from an auditing perspective is an interesting direction for future work  (\Cref{sec:discussion-conclusion}). \citet{kellerhals2024proportional} subsequently formalized metric variants of JR-like notions (including PJR) and explored connections
among various proportionality axioms; in particular, they showed that metric PJR (mPJR) is equivalent to PRF.

\paragraph{Verification and auditable fairness.}
The clustering literature discussed earlier in this section primarily focuses on designing algorithms that
produce fair outcomes or on establishing relationships between different fairness notions. In contrast, we
study the audit problem: given an arbitrary center selection, possibly produced
by a standard clustering method such as \(k\)-means or \(k\)-median, can one
efficiently certify whether it satisfies a proportional representation guarantee?
In approval-based multiwinner voting, verifying PJR is coNP-complete
\citep{aziz2018complexity}, while PJR+ admits polynomial-time verification
\citep{brill2023robust}; similar hardness results are known for other JR-style axioms, such as EJR \citep{aziz2017justified}, FJR and FPJR \citep{KalayciL025}. In metric clustering, verification has received much
less attention. Our work addresses this gap by (1) proving that mPJR verification is
coNP-complete, (2) defining mPJR+ as the metric analog of PJR+ and showing that it is
polynomial-time verifiable but computationally impractical at scale, and
(3) introducing DC-mPJR+, a structured relaxation that restricts attention to `tight'
coalitions, admits an \(O(mn\log n+mnk)\)-time verification algorithm, and remains
within a constant factor of mPJR+.
\section{Preliminaries}
\label{sec:preliminaries}
Let $(\mathcal{U}, d)$ be a metric space, where $\mathcal{U}$ is the universe of points and $d: \mathcal{U} \times \mathcal{U} \to \mathbb{R}_{\geq 0}$ is a distance function satisfying the properties of non-negativity, identity of indiscernibles, symmetry, and the triangle inequality.
For our algorithmic results, we assume that $d$ is computable in time $O(1)$.
 For any point $x \in \mathcal{U}$ and radius $r \ge 0$, we define the \textit{closed ball} $B(x, r)$ as the set of all points within distance $r$ from $x$, i.e.,
    \begin{equation}
        B(x, r) := \{y \in \mathcal{U} \mid d(x, y) \le r\}. \nonumber
    \end{equation}
%
We consider the problem of {\em centroid clustering}, i.e.,  selecting a set of $k$ representative centers for a multiset of $n$ datapoints in a metric space. 
Formally, an instance of centroid clustering in a metric space $({\mathcal U}, d)$ is a tuple $(N, M, k)$, where $N \subseteq \mathcal{U}$ is a multiset of $n$ datapoints, $M \subseteq \mathcal{U}$ is a set of $m$ candidate centers (or simply, \emph{candidates}), and $k\in [m]$ 
is the target number of clusters. The objective is to choose a \emph{center selection}, i.e., a subset of candidate centers $X \subseteq M$ with $|X| = k$. 
Following much of the literature on proportionality and fairness in centroid clustering, we will typically refer to the datapoints $N$ as \emph{agents}.  
For notational brevity, we let $q :=  n/k$ denote the {\em quota}: the minimum number of agents entitled to representation by one center in the selection.

\paragraph{Multiwinner Voting}
We build on prior works that link centroid clustering to multiwinner voting  with approval ballots. The multiwinner voting model is formally defined as follows.

\begin{definition}[Multiwinner voting]
An instance of {\em multiwinner voting with approval ballots} is given by a tuple $(N,M,(\mathcal A_i)_{i\in N},k)$, which consists of a set of agents (or voters) $N$, a set of candidates $M$, an approval profile $(\mathcal A_i)_{i\in N}$ giving the set of approved candidates $\mathcal A_i\subseteq M$ for each agent $i\in N$, and a target committee size~$k$. The goal is to select a winning set (a committee) $X\subseteq M$ of size $k$.
\end{definition}

There is a number of proportionality concepts for this setting that are based on the idea that any group of agents that is sufficiently large and has sufficiently cohesive preferences must be adequately represented in the committee. The requirement of cohesive preferences, typically interpreted as approval sets of the group members having a large common intersection, is critical to guaranteeing existence of outcomes satisfying the property. We will now introduce one such concept, namely, Proportional Justified Representation (PJR) \citep{sanchez2026proportional}. In what follows, we will explain why we chose PJR as our starting point.

\begin{definition}[PJR]
Given an instance $(N,M,(\mathcal A_i)_{i\in N},k)$ of multiwinner voting with approval ballots, a size-$k$ committee $X\subseteq M$ is said to satisfy {\em Proportional Justified Representation} (PJR) if for every integer $\ell\in[k]$ and every group of agents $S\subseteq N$ with $|S|\ge \ell\cdot q$
that satisfies $\left| \bigcap_{i \in S} {\mathcal A}_i \right| \ge \ell$ it holds that
$\left| X\cap \bigcup_{i \in S} {\mathcal A}_i \right| \ge \ell$.
\end{definition}

\paragraph{From Multiwinner Voting to Clustering}
\citet{kellerhals2024proportional} propose a general framework for translating notions of proportionality
from the setting of multiwinner voting with approval ballots to centroid-based clustering. To present this framework, it will be convenient to introduce the definition of a Group Approval Set: for any group of voters $S$ and a radius $r$, we define the approval set of $S$ at radius $r$ to consist of all candidates that lie at distance at most $r$ from at least one member of $S$.

\begin{definition}[Group Approval Set]
    Given an instance of centroid clustering $(N, M, k)$, the \emph{approval set} of a subset of agents $S\subseteq N$ at a radius $r\ge 0$, denoted $A_r(S)$, is the set of candidates within distance $r$ of at least one agent in $S$:
    \begin{equation}
        A_r(S) := \bigcup_{i\in S} \bigl(B(i,r)\cap M\bigr).
        \nonumber
    \end{equation}
    When $S=\{i\}$ is a singleton, we write $A_r(i)$ instead of $A_{r}(\{i\})$.
\end{definition}

Another important concept is that of $\ell$-cohesiveness: we say that a group of agents $S \subseteq N$ is \textit{$\ell$-cohesive at radius $r$} if $\left| \bigcap_{i \in S} A_r(i) \right| \ge \ell.$
In words, there are at least $\ell$ candidates that every member of $S$ would approve at radius $r$.

We are now ready to present the transformation proposed by \citet{kellerhals2024proportional}. Given an instance $(N, M, k)$ of centroid clustering in a metric space $({\mathcal U}, d)$, for each $r\in{\mathbb R}^+$ we construct an instance ${\mathcal I}_r = (N,M,(\mathcal A_i)_{i\in N},k)$ of  multiwinner voting with approval ballots by setting $\mathcal A_i=A_r(i)$ for each $i\in N$. Then, for a proportionality notion $\Pi$ that is defined for the setting of multiwinner voting with approval ballots, a center selection $X\subseteq M$ is said to satisfy {\em metric $\Pi$}, written as $m\Pi$, whenever $X$ satisfies $\pi$ in ${\mathcal I}_r$ for each $r\in{\mathbb R}^+$.

\citet{kellerhals2024proportional} argue that the notion of metric PJR (mPJR), defined by applying their transformation to PJR, is equivalent to the concept of \emph{Proportionally Representative Fairness} (PRF), introduced in an earlier paper by \citet{aziz2024proportionally}. 

\begin{definition}[PRF/mPJR]
    A center selection $X \subseteq M$ satisfies \textit{PRF/mPJR} if, for every radius $r \ge 0$, every integer $\ell \in [k]$, and every group of agents $S\subseteq N$ with $|S| \ge \ell \cdot q$ that is $\ell$-cohesive at radius $r$, 
    it holds that 
    $$|X \cap A_r(S)| \ge \ell.$$
\end{definition}

That is, if $\ell \cdot q$ agents all approve at least $\ell$ common candidates at radius $r$, then the center selection must include at least $\ell$ centers from their Group Approval Set $A_r(S)$.

An attractive feature of PRF (and hence mPJR) is that there is a polynomial-time algorithm that on every instance of centroid clustering admits a center selection that satisfies PRF \citep{aziz2024proportionally}. While in multiwinner setting there are more powerful axioms with this property, such as, e.g., EJR \citep{aziz2017justified}, prior to our work PJR was the strongest axiom for approval voting whose metric variant was known to be satisfiable for every instance.

In what follows, we will introduce alternative notions targeting proportional representation that are closely related to mPJR.

\section{The Intractability of Verifying Global Fairness}
\label{sec:hardness}

We now study the computational complexity of verifying whether a given center selection provides proportional representation guarantees. We show that verifying mPJR is coNP-complete (Section~\ref{sec:mpjr-hard}), introduce mPJR+ as a polynomial-time verifiable alternative (Section~\ref{sec:mpjr_plus}), show that the verification procedure for mPJR+ is impractical at scale, with a plausible route to improvement blocked by a coNP-hardness result (Section~\ref{sec:limits}), and describe a partial fix for small $k$ (Section~\ref{sec:fpt}): an algorithm whose running time can be bounded as $O(2^k(nk+m\log n))$.

\subsection{The coNP-Completeness of mPJR}\label{sec:mpjr-hard}

To show that verifying mPJR is coNP-complete, we reduce from the approval-based PJR verification problem, which is known to be coNP-complete \citep{aziz2018complexity}. 
The key tool is an embedding that translates instances of multiwinner voting with approval ballots into instances of centroid selection.

\begin{restatable}[Approval-to-metric embedding]{lemma}{lemApprovalMetric}
\label{lem:approval-metric}
Let $(N,M,(\mathcal A_i)_{i\in N},k)$ be an instance of multiwinner voting with approval ballots. 
One can construct, in polynomial time, a metric space $(\mathcal{U},d)$ with $\mathcal{U}=N\cup M$ and a polynomial-time computable distance function $d$ that takes values in $\{0, 1, 2, 3, 4\}$
so that
for every $S\subseteq N$ and every $1\le r<2$ it holds that
\[
\bigcap_{i\in S}(B(i,r)\cap M)=\bigcap_{i\in S} \mathcal A_i
\qquad\text{and}\qquad
A_r(S)=\bigcup_{i\in S} \mathcal A_i.
\]
\end{restatable}

The proof proceeds by constructing a weighted bipartite graph between agents and candidates, and assigning weight 1 to approval edges and weight 2 to dispproval edges; $d$ is then defined as the shortest-path metric. 
The full construction is given in Appendix~\ref{proof of lemma}.

This embedding preserves both cohesiveness and coverage at radius $r \in [1,2)$, so any PJR violation in the approval instance corresponds exactly to an mPJR violation in the metric instance, and vice versa. Combined with the coNP-completeness of approval-based PJR verification \citep{aziz2018complexity}, this yields the following theorem (see Appendix~\ref{proof of thm} for proof details).

\begin{restatable}{theorem}{thmMpjrHard}\label{thm:mpjr-hard}
Checking whether a given center selection satisfies mPJR is coNP-complete.
\end{restatable}

Consequently, unless P=NP, there is no efficient algorithm that can verify whether an arbitrary center selection satisfies mPJR. 
This motivates the search for an alternative fairness axiom that retains the spirit of proportional representation while admitting efficient verification.

\subsection{A Metric Analog of PJR+}
\label{sec:mpjr_plus}

In the setting of multiwinner voting with approval ballots, PJR is coNP-complete to verify; several stronger properties, such as EJR, FJR, and FPJR, are hard to verify as well \citep{aziz2017justified,KalayciL025}.
In part to circumvent these hardness results, \citet{brill2023robust} introduced PJR+ as a strengthening of PJR and showed that it admits polynomial-time verification. 
The key structural insight behind PJR+ is to anchor each potential violation to a specific unselected candidate, rather than quantifying over arbitrary cohesive groups.
Formally, PJR+ is defined as follows.

\begin{definition}[PJR+]
Given an instance $(N,M,(\mathcal A_i)_{i\in N},k)$ of multiwinner voting with approval ballots, a size-$k$ committee $X\subseteq M$ is said to satisfy PJR+ if for every integer $\ell\in[k]$, every unselected candidate $c\in M\setminus X$ and every group of agents $S\subseteq N$ with $|S|\ge \ell\cdot q$ such that $c\in {\mathcal A}_i$ for all $i\in S$
 it holds that
$\left| X\cap \bigcup_{i \in S} {\mathcal A}_i \right| \ge \ell$.
\end{definition}

We adapt this idea to the metric setting using the Kellerhals--Peters transformation, and refer to the resulting notion as mPJR+. 
Intuitively, where mPJR asks whether any sufficiently large and cohesive group lacks representation,
mPJR+ asks whether any unselected center $c$ witnesses a violation.
As a result, a certificate of violation becomes a single center and a group of agents, as opposed to a set of centers and a group of agents.
As we will see, even in the metric setting, this change suffices to admit polynomial-time verification via submodular minimization. 
Our definition of mPJR+ is parameterized by a multiplicative approximation factor, which relaxes the distance between agents in $S$ and their closest representatives in $X$.  

\begin{definition}[$\gamma$-mPJR+]\label{def:mpjrplus}
    Given an instance $(N, M, k)$ of centroid clustering, a center selection $X \subseteq M$ is said to satisfy $\gamma$-mPJR+ for $\gamma \geq 1$ if 
    for every integer $\ell\in[k]$, every unselected center $c\in M\setminus X$, every group of agents $S\subseteq N$ with $|S|\ge \ell\cdot q$, and for $r(c, S) := \max_{i \in S} d(i,c)$
    it holds that
    \[
    |X \cap A_{\gamma\cdot r(c, S)}(S)| \ge \ell.
    \]
     When $\gamma=1$, we simply write mPJR+.
\end{definition}

\begin{remark}
It is not immediate that \Cref{def:mpjrplus} (with $\gamma=1$) is obtained by following the template of \citet{kellerhals2024proportional}: 
a  direct application of their transformation results in an axiom that makes the same requirement, but for every radius $r\geq 0$ (instead of $r(c, S)$) and with the additional requirement that $d(i,c)\leq r$ for each $i\in S$. To argue that the two axioms are equivalent, we observe that a certificate of a violation of the latter axiom for some radius $r\geq 0$, candidate $c$, group $S$, and integer $\ell$ also constitutes a violation with a radius value of $r(c, S)$. 
To see this, note that 
if $d(i,c)\leq r$ for all $ i\in S$ then $r\geq r(c, S)$, and thus $|X \cap A_{r(c, S)}(S)| \leq |X \cap A_r(S)| < \ell$. 
A formal equivalence proof is given in Appendix~\ref{app:kp-template}.
\end{remark}

%
\begin{remark}
As approval-based PJR+ implies approval-based PJR, the equivalence above together with the Kellerhals--Peters lifting template also imply that mPJR+ implies mPJR. For completeness, we give a direct argument. Consider a violation of mPJR, witnessed by a group $S$ that is $\ell$-cohesive at radius $r$ with $|X \cap A_r(S)| < \ell$. Note that all candidates in \(\bigcap_{i\in S} A_r(i)\) belong to
\(A_r(S)\) and \(|X\cap A_r(S)|<\ell\), so at least one candidate in \(\bigcap_{i\in S} A_r(i)\) is unselected. 
Choose some such candidate \(c\in \bigcap_{i\in S} A_r(i) \cap (M\setminus X)\). 
Then \(\max_{i\in S}d(i,c)\le r\), so \(c\) serves as an anchoring center for an mPJR+ violation.
\end{remark}

Since mPJR is the strongest proportional representation criterion known to always be satisfiable in centroid-based clustering, a natural first question is whether mPJR+ is always satisfiable. 
We will now answer this question in the affirmative by showing that mPJR+ is satisfied by the Spatial Expanding Approval Rule (SEAR) \citep{aziz2024proportionally} --- the same (polynomial-time) algorithm used to show satisfiability of mPJR.
At a high level, SEAR computes a center selection by initializing each agent with a budget of $1$ 
and uniformly growing balls centered at the candidates.
It iteratively selects candidates when their ball contains agents whose collective budget is at least $q=n/k$ 
and updates the individual budgets of these agents so that their joint budget is reduced by $q$. 

\begin{restatable}{proposition}{propSearPjrplus}
\label{prop:sear-pjrplus}
    The center selection returned by SEAR always satisfies mPJR+.
\end{restatable}
 
The proof is a direct adaptation of the argument that SEAR satisfies mPJR \citep{aziz2024proportionally}, and can be found in Appendix~\ref{proof of propSEAR}.

Having established satisfiability, we now turn to the verification complexity of mPJR+. We show that, unlike mPJR, the new mPJR+ property is polynomial-time verifiable.
Our verification algorithm adapts the argument of \citet{brill2023robust} for PJR+ to the metric setting.

\begin{restatable}{proposition}{propPjrplusCheck}
\label{prop:pjrplus-check}
    Given an instance $(N, M, k)$ of centroid clustering and a center selection $X\subseteq M$, it can be verified in polynomial time whether $X$ satisfies mPJR+.
\end{restatable}
\begin{proof}[Proof sketch] 
Fix $c \in M \setminus X$ and $r \geq 0$, and consider agents in $B(c,r) \cap N$. For $S \subseteq B(c,r) \cap N$, define $f_{c,r}(S) := |X \cap A_r(S)| - |S|/q$. The function $f_{c,r}$ is submodular (as $S \mapsto |X \cap A_r(S)|$ is a coverage function). An mPJR+ violation exists at $(c,r)$ if and only if $\min_S f_{c,r}(S) \leq -1$. Since only radii in $\{d(i,c) : i \in N\}$ need to be checked, verification reduces to $O(mn)$ submodular minimizations, each solvable in polynomial time. See Appendix~\ref{proof of pjrplus-check} for the full proof.
\end{proof}

\subsection{Complexity of mPJR+ Verification}\label{sec:limits}
Together, Propositions~\ref{prop:sear-pjrplus} and~\ref{prop:pjrplus-check} establish mPJR+ as a theoretically appealing alternative to mPJR. 
However, the time complexity of the verification procedure presented in the proof of Proposition~\ref{prop:pjrplus-check} is far too high to be practical.


\begin{remark}	\label{rmk:mpjrplus-check-complexity}
In the worst case, the construction in the proof of Proposition~\ref{prop:pjrplus-check} requires solving a submodular minimization problem $\mathcal{O}(m n)$ times. 
Thus, even using the state-of-the-art algorithms for submodular minimization results in an $\mathcal{O}(m n^4 k \log{nk} + m n^4 \log^{O(1)}{nk})$ weakly polynomial time algorithm or an $\mathcal{O}(m n^5 k \log^2{n} + m n^5 \log^{O(1)}{n})$ strongly polynomial time algorithm \citep{lee2015faster}.\footnote{\citet{lee2015faster} give a strongly polytime algorithm which runs in time $\mathcal{O}(EO\cdot n^3\log^2n + n^4\log n)$ and a weakly polytime algorithm which runs in time $\mathcal{O}(n^2 \log nM \cdot EO + n^3\log nM)$, where $n$ is the ground set over which subsets are taken ($N_c$ in our case), $M$ is the largest absolute value of the submodular function ($k$ in our case), and $EO$ is the time for one evaluation-oracle call. In our case, we can precompute, for each $n\cdot m$ radii of interest, the set of agents within that radius of each candidate, but we will still need to look for an agent in $N'$ in each of these lists for each candidate in $X$, resulting in an evaluation-oracle time complexity of $\mathcal{O}(k\cdot n),$ and yielding the bounds given in \Cref{rmk:mpjrplus-check-complexity}.}
\end{remark}

A natural approach toward a more efficient algorithm is to avoid submodular minimization and instead decompose the verification problem, by fixing the representation level $\ell$ and looking for violations of mPJR+ that are witnessed by cohesive groups of size $\ell\cdot q$ which are close to fewer than $\ell$ selected centers. Indeed, for the related (but more demanding) concept of EJR+ \citep{brill2023robust}, this approach leads to a simple polynomial-time verification algorithm (we note, however, that mEJR+ --- the metric variant of EJR+ ---
is known to be unsatisfiable\footnote{More precisely, mEJR is known to be unsatisfiable \citep{aziz2024proportionally}, and EJR+ is a strengthening of EJR; accordingly, mEJR+ is a strenthening of mEJR and is therefore unsatisfiable as well.}, and hence mEJR+ does not fulfill our criteria of a `good' metric proportionality axiom). However, it turns out that this strategy does not work for (m)PJR+: 
strikingly, we will now prove that the problem of verifying the fixed-$\ell$ version of (m)PJR+ is coNP-complete. 

 Formally, given a metric space $({\mathcal U},d)$, an instance $(N, M, k)$ of centroid clustering in $({\mathcal U},d)$, a size-$k$ center selection
    $X\subseteq M$, and an integer $\ell\in[k]$, we say that $X$ satisfies
    {\em fixed-$\ell$ mPJR+} if for every unselected center $c\in M\setminus X$, 
    every group of agents $S\subseteq N$ with $|S|\ge \ell\cdot q$, and $r(c, S) =\max_{i\in S}d(i,c)$, we have
        $|X\cap A_{r(c, S)}(S)|\ge \ell.$
    
\begin{restatable}{proposition}{propMpjrplusFixedell}
\label{prop:mpjrplus_fixedell}
    Given a metric space $({\mathcal U},d)$, an instance $(N, M, k)$ of centroid clustering in $({\mathcal U},d)$, a size-$k$ center selection
    $X\subseteq M$, and an integer $\ell\in[k]$, 
    deciding
    whether $X$ satisfies fixed-$\ell$ mPJR+ is coNP-complete. 
\end{restatable}

\begin{proof}[Proof sketch]
    Membership in coNP is immediate:
    a violation of fixed-$\ell$ mPJR+ is certified by an unselected center $c\in M\setminus X$ and a coalition of agents $S\subseteq N$.
    %
    It remains to prove coNP-hardness.
    We consider the complement of our problem, i.e., deciding whether $X$ violates fixed-$\ell$ mPJR+, and prove it to be NP-hard.
    
    We proceed by first establishing an analogous result in the setting of multiwinner voting with approval ballots, and then employ Lemma~\ref{lem:approval-metric} to port this result to the metric setting.
    The approval-based result, given in Appendix~\ref{app:discrete_hardness}, uses a reduction from \textsc{Balanced Biclique} \citep{GareyJohnson1979}. 
    Given a bipartite graph $G = (L, R, E)$ and a target biclique size~$t$, we construct an approval instance in which voters correspond to vertices in $L$, candidates correspond to vertices in $R$ plus one additional unselected candidate $z$, voter $v_u$ approves candidate $c_w$ if and only if $\{u,w\} \notin E$, and $z$ is approved by all voters. 
    The committee is set to $X = \{c_w : w \in R\}$ and the representation level to $\ell = t$. 
    Under this construction, a coalition $S$ of size $t$ whose members collectively approve fewer than $t$ committee members (and approve $z$) corresponds exactly to a set of $t$ vertices in $L$ that are \emph{all} adjacent to some common set of $t$ vertices in $R$ --- that is, a balanced biclique in $G$.
\end{proof}

As noted in the proof sketch above, our proof proceeds by deriving a hardness result for the setting of multiwinner voting with approval ballots. Since we are primarily interested in the metric setting, we defer the formal statement of this result to Appendix~\ref{app:discrete_hardness}, but note that it may be of independent interest. 

\subsection{A Workaround for Small $k$}\label{sec:fpt}

On the positive side, there is a simple algorithm for verifying mPJR+ whose running time is exponential in $k$, but scales gracefully with $n$ and $m$.

\begin{proposition}\label{prop:direct-mpjrplus-verifier}
Let $X \subseteq M$ be a center selection of size $k$. For every unselected center $c \in M \setminus X$, every radius $r \in \{d(i, c) \mid i \in N\}$, and every subset $Y \subsetneq X$, define
\[
S(c, r, Y) := \{i \in N \mid d(i, c) \le r \text{ and } d(i, x) > r \text{ for all } x \in X \setminus Y\}.
\]
Then $X$ violates mPJR+ if and only if there exist a triple $(c, r, Y)$ with $|S(c, r, Y)| \ge (|Y| + 1)\cdot q$. This condition can be verified in time $O(mn\log n\cdot 2^k)$. 
\end{proposition}

\begin{proof}
\emph{($\Rightarrow$)} Suppose $X$ violates mPJR+, witnessed by $c \in M \setminus X$, $S \subseteq N$, and $\ell \in [k]$ with $|S| \ge \ell\cdot q$ and $|X \cap A_r(S)| < \ell$ for $r := \max_{i \in S} d(i, c)$. Set $Y := X \cap A_r(S)$; then $|Y| < \ell$, so $\ell \ge |Y| + 1$, and $r = d(i, c)$ for some $i \in S$, so $r \in \{d(i, c) : i \in N\}$. We claim that $S \subseteq S(c, r, Y)$: indeed, for every $i \in S$ we have $d(i, c) \le r$, and for every $x \in X \setminus Y$ we have $x \notin A_r(S)$, so $d(i, x) > r$. Therefore $|S(c, r, Y)| \ge |S| \ge \ell\cdot q \ge (|Y| + 1)\cdot q$.

\emph{($\Leftarrow$)} Conversely, suppose $|S(c, r, Y)| \ge (|Y| + 1)\cdot q$ for some $c, r, Y$. Set $S' := S(c, r, Y)$, $\ell := |Y| + 1 \in [k]$ (since $Y \subsetneq X$), and $r' := \max_{i \in S'} d(i, c) \le r$. For every $x \in X \setminus Y$ and every $i \in S'$ we have $d(i, x) > r \ge r'$, so $X \cap A_{r'}(S') \subseteq Y$ and thus $|X \cap A_{r'}(S')| \le |Y| < \ell$. Combined with $|S'| \ge \ell\cdot q$, this witnesses an mPJR+ violation at $(c, S', \ell)$.

This characterization immediately implies a bound of $O(m \cdot n \cdot 2^k \cdot nk) = O(m n^2 k\, 2^k)$ on the running time. Indeed, there are at most $m$ choices of $c$, at most $n$ relevant radii per $c$, and $2^k$ choices of $Y \subsetneq X$. For each triple, $|S(c, r, Y)|$ is computed by scanning $n$ agents and checking distances to at most $k$ selected centers, in $O(nk)$ time. 

To obtain the stated bound on the running time, we use a slightly different approach (see Algorithm~\ref{alg:verify-mpjr-smallk} for the pseudocode). We loop over all sets $Y\subseteq X$ ($2^k$ possibilities). For each set $Y\subseteq X$, we go over all agents $i\in N$, compute $u_i=\min_{c\in X\setminus Y} d(i, c)$ (time $O(nk)=O(nm)$ in total), and then loop over all candidates $c\in M\setminus X$ ($m$ options). Within this inner loop, for each $i\in N$ we set $\ell_i=d(i, c)$; then, checking if our condition holds for the given $Y$ and $c$ and some $r\in\mathbb R$ reduces to deciding if there is some $r\in\mathbb R$ that belongs to at least $(|Y|+1)\cdot q$ segments in the collection $\{[\ell_i, u_i)\mid i\in N\}$: indeed,  $r\in[\ell_i, u_i)$ for some $i\in N$ if and only if $d(i, c)\le r$, but $d(i, j)>r$ for all $j\in X\setminus Y$.

The problem of finding the maximum intersection of $n$ intervals on the real line admits an $O(n\log n)$ algorithm \citep{KT}: one can sort all $2n$ endpoints, initialize a counter at 0 and sweep through the sorted list, incrementing the counter whenever we encounter a left endpoint and decrementing it whenever we encounter a right endpoint. Since the intervals in our collections are half-closed, whenever a point $r$ is the right endpoint of one interval and the left endpoint of another interval, i.e., $r=\ell_i=u_j$ for some $i, j\in N$, we need to decrement the counter at $r$ before we increment it (see the pseudocode). 
\end{proof}

\begin{algorithm}[h]
\caption{Verifying mPJR+ in time $O(mn\log n\cdot 2^k)$}
\label{alg:verify-mpjr-smallk}
\KwIn{An instance of centroid clustering $(N, M, k)$ in a metric space $(\mathcal U,d)$ with $|N|=n$, $|M|=m$, and a center selection $X\subseteq M$}
\KwOut{\textsc{True} if $X$ satisfies mPJR+; otherwise \textsc{False}}

\ForEach{$Y\subsetneq X$}{
    \ForEach{$i\in N$}{

        $u_i\gets +\infty$

        \ForEach{$j\in X\setminus Y$}{
            $u_i\gets \min\{u_i, d(i, j)\}$ 
        }

    }
    
    \ForEach{$c \in M \setminus X$}{
    
    $\ell_i\gets d(i, c)$
    
    Create a multiset $R\gets \{(\ell_i, \textsf{L}), (u_i, \textsf{U})\mid i\in N, \ell_i<u_i\}$
    
    Sort the pairs in $R$ in non-decreasing order of the first element, breaking ties so that whenever $R$ contains both $(r, \textsf{L})$
    and $(r, \textsf{U})$ for some $r$, the pair $(r, \textsf{U})$
    comes before $(r, \textsf{L})$ in the sorted order
     
    Initialize $s \gets 0$, $s_{\max}=0$
    
    \ForEach{$(r, \textsf{Z})\in R$ in the sorted order}{
        \If{$\textsf{Z}=\textsf{L}$}{
        
        $s\gets s+1$

        \If{$s_{\max}<s$}{$s_{\max}\gets s$}
        }
        \If{$\textsf{Z}=\textsf{U}$}{

        $s\gets s-1$
        
        }
    }
      
    \If{$s_{\max}\ge (|Y|+1)\cdot n/k$}{\Return \textsc{False}}
    }

}

\Return \textsc{True}\;
\end{algorithm}

\begin{remark}
    We can modify Algorithm~\ref{alg:verify-mpjr-smallk} to first create $m$ lists of agents, where in the $j$-th list the agents are sorted by distance from the $j$-th candidate, and then, when building a sorted list of endpoints of $2n$ segments for a given $Y$ and $c$, merge the $k+1$ sorted lists that corresponds to candidates in $Y\cup\{c\}$ and select the relevant elements from the merged list, instead of sorting the $2n$ elements from scratch. The running time of the modified algorithm is 
    $O(mn\log n + mnk\log k\cdot 2^k)$; for $k\log k  < \log n$, this offers an improvement over Algorithm~\ref{alg:verify-mpjr-smallk}.
\end{remark}

\begin{remark}
   It is straightforward to adapt Algorithm~\ref{alg:verify-mpjr-smallk} to the approval setting, i.e., to the problem of verifying whether a given committee satisfies PJR+. In the approval setting, for a given $Y\subsetneq X$ we only need to identify the voters that do not approve any candidate in $X\setminus Y$, and then, for each $c\in M\setminus X$, check how many of these voters approve $c$. Thus, the running time of the algorithm improves to $O(2^k(nk+nm))=O(nm\,2^k)$.      
\end{remark}

\begin{remark}
    One can view \Cref{prop:direct-mpjrplus-verifier} as a fixed-parameter tractability result with respect to $k$ (of course, technically speaking, the verification problem for mPJR+ is in FPT for any parameter, as it is, in fact, in P by \Cref{prop:pjrplus-check}). We note that one cannot hope for a similar FPT result for mPJR. This is because hardness of PJR verification is shown by a reduction from {\sc Balanced Biclique}, 
    which maps the problem of deciding whether a bipartite graph admits an $\ell$-by-$\ell$ biclique to the problem of deciding whether a size-$(2\ell-2)$ committee fails PJR \citep{aziz2018complexity}. As {\sc Balanced Biclique} is W[1]-hard with respect to $k$ \citep{lin}, the reduction of \citet{aziz2018complexity} shows that verifying PJR is co-W[1]-hard; this result extends to mPJR via \Cref{lem:approval-metric}. 
    Interestingly, however, even though the proof of \Cref{prop:mpjrplus_fixedell}, too, proceeds by a reduction from {\sc Balanced Biclique}, the reduction in that proof is not parameter-preserving, and, indeed, the problem of deciding if a given committee satisfies fixed-$\ell$ (m)PJR+ is in FPT with respect to $k$. Indeed,  we can modify  Algorithm~\ref{alg:verify-mpjr-smallk} to only consider sets $Y$ of size at most $\ell-1$, and reject if we find some such set $Y$ with $|S(c, r, Y)|\ge \ell$.
\end{remark}

\section{Efficient Verification via Default Coalitions}
\label{sec:sc_mpjr_plus}
We showed in the previous section that mPJR is coNP-complete to verify, and while mPJR+ is polynomial-time verifiable in theory, the reliance on repeated submodular minimization makes verification impractical at scale. \Cref{prop:direct-mpjrplus-verifier} provides a solution that is suitable for small $k$, but not in general.
In this section, we resolve this tension by introducing a new proportional representation criterion, \emph{Default Coalitions mPJR+} (DC-mPJR+), that is both efficiently verifiable and provably close to mPJR+.

The key idea is to replace the search over arbitrary coalitions of agents with a search over a structured family of coalitions, one for each unselected center.
We define these coalitions and the resulting axiom in Section~\ref{sec:canonical}.
While this new formulation sacrifices generality in spirit, we are able to show that the resulting axiom implies a close approximation to mPJR+ in Section~\ref{sec:global}, and give an $O(mn\log n + mnk)$ verification algorithm in Section~\ref{sec:algorithm}.

\subsection{Default Coalitions and the DC-mPJR+ Axiom}\label{sec:canonical}

Recall that verifying mPJR+ is computationally expensive, as one must search over all possible coalitions $S$ that could witness a violation with respect to an unselected center $c$. 
The coalition $S$ can be any subset of agents near $c$, and the violation radius $r(c, S) = \max_{i \in S} d(i,c)$ depends on the choice of~$S$. 

Our approach is to eliminate this freedom by fixing, for each unselected center $c$ and each representation level $\ell$, a single coalition to check. 
Specifically, we focus on the coalition most naturally forming around the candidate $c$: the $\ell\cdot q$ agents closest to $c$.

\begin{definition}[Default Coalitions]
    For any unselected center $c \in M \setminus X$ and representation level $\ell \in [k]$, let $r_{c,\ell}$ denote the minimum radius required for a closed ball centered at $c$ to capture enough agents to justify $\ell$ representatives, i.e.,
    \begin{equation}
        r_{c,\ell} := \min \{ r \ge 0 \mid |B(c, r) \cap N| \ge \ell \cdot q \}.
    \end{equation}
    We refer to the set of agents captured within this ball, denoted by $N_{c,\ell}$, as the {\em default coalition} for candidate $c$ at level $\ell$, i.e.,
    \begin{equation}
        N_{c,\ell} := B(c, r_{c,\ell}) \cap N.
    \end{equation}
\end{definition}
%
Intuitively, among all groups of $\ell\cdot q$ or more agents around $c$, the default coalition $N_{c,\ell}$ 
is likely the densest and hence has the strongest claim to representation near $c$.

We now define our proportional representation criterion by restricting the mPJR+ check to default coalitions only.
\begin{definition}[$\gamma$-Default Coalitions mPJR+] \label{def:cc-mpjrplus}
    Given an instance $(N, M, k)$ of centroid clustering, a center selection $X \subseteq M$ is said to satisfy {\em $\gamma$-Default Coalitions mPJR+ ($\gamma$-DC-mPJR+)} for $\gamma\geq 1$ if there is no unselected center $c \in M \setminus X$ and representation level $\ell \in [k]$ such that
    \begin{equation}
        |X \cap A_{\gamma \cdot r_{c,\ell}}(N_{c,\ell})| < \ell.
    \end{equation}
    When $\gamma=1$, we simply write DC-mPJR+.
\end{definition}

Since $\gamma$-DC-mPJR+ only checks the $\gamma$-mPJR+ condition for a specific coalition $S = N_{c,\ell}$ at each pair $(c, \ell)$, it is formally a weakening of $\gamma$-mPJR+. 
Thus, any center selection satisfying $\gamma$-mPJR+ also satisfies $\gamma$-DC-mPJR+, and in particular, SEAR satisfies DC-mPJR+ by \Cref{prop:sear-pjrplus}.

The total number of pairs $(c, \ell)$ to check is at most $mk$, compared to the exponentially many coalitions required by mPJR or the $O(mn)$ submodular minimizations required by mPJR+. 
This compression of the search space is what makes efficient verification possible;
we will make this argument formal in \Cref{sec:algorithm}.

\paragraph{Relationship to mPJR.}
Since DC-mPJR+ restricts the mPJR+ condition to default coalitions, it follows that mPJR+ implies DC-mPJR+, in the sense that if a center selection satisfies mPJR+, it also satisfies DC-mPJR+. Moreover, in \Cref{sec:mpjr_plus}
we have argued that mPJR+ implies mPJR. It is natural to ask, then, what is the relationship between DC-mPJR+ and mPJR. It turns out that neither implies the other.

\begin{proposition}\label{prop:dc-mpjr-incomparable}
DC-mPJR+ and mPJR are incomparable.
\end{proposition}

\begin{proof}
We construct two metric instances. In both, $N = \{1, \dots, 6\}$, $k = 3$, $q = n/k = 2$, and all agent--candidate distances are either $1$ or $2$;
to complete the description of each instance, 
we specify the set of candidates $M$, a center selection $X$, and, for each agent $i\in N$, the set $D_i := \{c \in M \mid d(i, c) = 1\}$ (so that $d(i, c)=2$ for all $c\in M\setminus D_i$).

\medskip
\noindent\emph{Instance 1: DC-mPJR+ holds, mPJR fails.}
Let $M = \{a, b, x_1, x_2, x_3\}$, $X = \{x_1, x_2, x_3\}$, and
\[
D_i = \{a, b, x_1\} \text{ for } i \in \{1,2,3,4\},
\quad
D_5 = \{a, b, x_2\},
\quad
D_6 = \{a, b, x_3\}.
\]
\begin{itemize}
\item
\emph{mPJR fails:} Take $S = \{1, 2, 3, 4\}$, $\ell = 2$, $r = 1$. Then $|S| = \ell\cdot q$ and $\bigcap_{i \in S} (B(i, 1) \cap M) = \{a, b, x_1\}$, so $S$ should get $\ell$ representatives at radius $1$. However, $A_1(S) = \{a, b, x_1\}$, so $|X \cap A_1(S)| = 1 < 2$.
\item
\emph{DC-mPJR+ holds:} The unselected candidates are $a$ and $b$, and every agent lies at distance $1$ from both. Hence for each $c \in \{a, b\}$ and each $\ell \in [3]$, we have $r_{c, \ell} = 1$ and $N_{c, \ell} = N$, so $A_1(N_{c, \ell}) = M$ and $|X \cap A_1(N_{c, \ell})| = 3 \ge \ell$.
\end{itemize}

\medskip
\noindent\emph{Instance 2: mPJR holds, DC-mPJR+ fails.}
Let $M = \{z, x_1, x_2, x_3\}$, $X = \{x_1, x_2, x_3\}$, and
\[
D_i = \{z, x_1\} \text{ for } i \in \{1,2,3\},
\quad
D_4 = \{z\},
\quad
D_5 = \{x_2\},
\quad
D_6 = \{x_3\}.
\]
\begin{itemize}
\item
\emph{DC-mPJR+ fails:} The only unselected candidate is $z$. For $\ell = 2$, the radius-$1$ ball around $z$ contains exactly $\{1, 2, 3, 4\}$, so $r_{z, 2} = 1$ and $N_{z, 2} = \{1, 2, 3, 4\}$. But $A_1(N_{z, 2}) = \{z, x_1\}$, giving $|X \cap A_1(N_{z, 2})| = 1 < 2$.
\item
\emph{mPJR holds:} Since all distances are $1$ or $2$, it suffices to consider the radii $r$ with $1\le r<2$: indeed,  for $r < 1$, no agent has any candidate within distance $r$, so there are no $\ell$-cohesive groups for $\ell>0$, whereas for $r \ge 2$ we have $A_r(i)=M$ for all $i\in N$, so $|X \cap A_r(S)| = k \ge \ell$ for every group $S\subseteq M$. For $1 \le r < 2$ we have $B(i, r) \cap M = D_i$ for every $i\in N$. 
Note that for $\ell\ge 2$ there is no group of $\ell\cdot q$ agents that is $\ell$-cohesive at radius $r$. For $\ell=1$, every $\ell$-cohesive group at radius $r$ is a size-$2$ subset of $\{1, 2, 3, 4\}$, and in any such subset at least one agent approves $x_1\in X$ at radius $1$.
Therefore, mPJR holds. 
\end{itemize}
Together, the two instances prove that mPJR and DC-mPJR+ are incomparable.
\end{proof}

\subsection{DC-mPJR+ Implies Approximation to mPJR+}
\label{sec:global}

A natural concern is whether restricting to default coalitions sacrifices too much: could a solution satisfy DC-mPJR+ while badly failing proportional representation for some group of agents? 
The following result shows that this cannot happen: any solution satisfying $\gamma$-DC-mPJR+ also satisfies $(\gamma+2)$-mPJR+.
The proof uses the triangle inequality to relate the distance from an arbitrary agent to a selected center via the default coalition. The additive gap of 2 arises from two intermediate hops: one from the agent to the unselected center $c$, and one from $c$ to a member of the default coalition.

\begin{restatable}{theorem}{thmScMpjrplusImpliesMpjrplus}
\label{prop:sc-mpjrplus-implies-mpjrplus}
If a center selection satisfies $\gamma$-DC-mPJR+ for some $\gamma\geq 1$, then it also satisfies $(\gamma + 2)$-mPJR+.
\end{restatable}
\begin{proof}
Consider an instance $(N, M, k)$ of centroid clustering in a metric space $({\mathcal U}, d)$, and
a center selection $X\subseteq M$ that satisfies $\gamma$-DC-mPJR+ for some $\gamma\geq 1$. Fix $c\in M\setminus X$ and $\ell\in[k]$. By $\gamma$-DC-mPJR+, we know that $|X\cap A_{\gamma\cdot r_{c,\ell}}(N_{c,\ell})|\geq \ell$. That is, there is a size-$\ell$ subset $X'\subseteq X$ such that for each $x\in X'$ there is at least 
one agent $i\in N_{c,\ell}$ with
$d(i, x)\le \gamma \cdot r_{c,\ell}$.

Now fix an arbitrary group of agents $S\subseteq N$ such that $|S|\geq \ell\cdot q$ and let $r(c, S) =\max_{i\in S} d(i, c)$. 
Note that, by the choice of $N_{c, \ell}$, we have $r_{c, \ell}\le r(c, S)$: indeed, of all groups of size at least $\ell\cdot q$ (including $S$), the group $N_{c, \ell}$ is the `tightest'. 
Let $x$ be an arbitrary center in $X'$ and let $j\in N_{c,\ell}$ be some agent such that $d(x,j)\leq \gamma \cdot r_{c,\ell}$. 

Now consider the distance from an arbitrary agent $i\in S$ to the selected center $x$:
\begin{align*}
    d(i,x) &\leq d(i,c) + d(c,j) + d(j,x) \\
    &\leq r(c, S) + r_{c,\ell} + \gamma\cdot r_{c,\ell} \\
    &\leq (\gamma + 2) \cdot r(c, S).
\end{align*}
Since this inequality holds for every $x\in X'$, it follows that $|X\cap A_{(\gamma + 2)\cdot r(c, S)}(S)| \geq |X'| = \ell$. Since this holds for arbitrary unselected center $c\in M\setminus X$, group of agents $S\subseteq N$, and $\ell\in[k]$, we conclude that $X$ satisfies $(\gamma+2)$-mPJR+.
\end{proof}

\Cref{prop:sc-mpjrplus-implies-mpjrplus} means that restricting to default coalitions costs very little in terms of fairness guarantees, while --- as we are about to show --- yielding a dramatic improvement in verification time complexity.

\subsection{Direct Polynomial-Time Verification}\label{sec:algorithm}

We now give an efficient algorithm for verifying $\gamma$-DC-mPJR+ (Algorithm~\ref{alg:verify-sc-mpjrplus}). 

At a high level, we grow a ball around each unselected center $c$ by processing agents in batches, in order of increasing distance from $c$. 
We maintain, for each selected center $x \in X$, the quantity $\delta(x)=\min_{i\in P} d(x,i)$, which tracks the distance from $x$ to the current prefix $P$ of agents. 
Whenever we update $P$, we compute 
the representation level it deserves, i.e., 
$t=\lfloor |P|/q\rfloor$.
We then compute the number of selected centers that lie at distance at most $\gamma\cdot\max_{i\in P}d(c, i)$ from $P$; 
if this number is smaller than $t$, then we have found a violation. The correctness of the algorithm follows from the fact that each set $N_{c, \ell}$ corresponds to some prefix $P$ considered by the algorithm.

\begin{algorithm}[h]
\caption{Verifying $\gamma$-DC-mPJR+}
\label{alg:verify-sc-mpjrplus}
\KwIn{An instance of centroid clustering $(N, M, k)$ for a metric space $(\mathcal U,d)$ with $|N|=n$, $|M|=m$, a center selection $X\subseteq M$, and parameter $\gamma \ge 1$}
\KwOut{\textsc{True} if $X$ satisfies $\gamma$-DC-mPJR+; otherwise \textsc{False}}


\ForEach{$c \in M \setminus X$}{
    Sort the agents in non-decreasing order of distance from $c$: 
    $d(c,i_1)\le d(c,i_2)\le \cdots \le d(c,i_n)$
    
    Let $\rho_1<\rho_2<\cdots<\rho_s$ be the distinct values among 
    $\{d(c,i):i\in N\}$
    
    Initialize $P \gets \varnothing$
    
    \ForEach{$x\in X$}{
        $\delta(x)\gets +\infty$
    }
    
    
    \For{$j=1$ \KwTo $s$}{
        Add to $P$ all agents $i\in N$ such that $d(c,i)=\rho_j$
        
        \ForEach{newly added agent $i$}{
            \ForEach{$x\in X$}{
                $\delta(x)\gets \min\{\delta(x),d(x,i)\}$
            }
        }
        
        Let $t\gets \left\lfloor |P|\cdot k/n\right\rfloor$
        
            Compute $\mathrm{cov}\gets |\{x\in X:\delta(x)\le \gamma\cdot  \rho_j\}|$
            
            \If{$\mathrm{cov}<t$}{
                \Return \textsc{False}
            }
            
    }
}
\Return \textsc{True}\;
\end{algorithm}

\begin{restatable}{theorem}{thmScMpjrplusCheck}
\label{prop:sc-mpjrplus-check}
Algorithm~\ref{alg:verify-sc-mpjrplus} correctly decides whether a center selection $X$ satisfies $\gamma$-DC-mPJR+. Its running time is $O(mn\log n + mnk)$.
\end{restatable}
\begin{proof}
We first prove correctness. According to \Cref{def:cc-mpjrplus}, $X$ satisfies $\gamma$-DC-mPJR+ if and only if for every $c\in M\setminus X$ and every $\ell\in[k]$ it holds that $|X\cap A_{\gamma\cdot r_{c,\ell}}(N_{c,\ell})|\ge \ell$. 
Now fix $c\in M\setminus X$, and consider the iteration of the outer loop in which $c$ is processed. 
The agents are sorted by distance from $c$, and the algorithm scans the distinct radii \(\rho_1<\rho_2<\cdots<\rho_s\). 

Consider step $j$ of the inner {\bf for} loop. It holds that \(P=B(c,\rho_j)\cap N\). Hence, for any $\ell\in [k]$, whenever $|P|\ge \ell\cdot q$ for the first time, we have \(\rho_j=r_{c,\ell} \) and $P=N_{c,\ell}$.
Moreover, for every center $x\in X$, the maintained value $\delta(x)=\min_{i\in P} d(x,i)$ is exactly the distance from $x$ to $P$. Therefore, $x\in A_{\gamma\cdot r_{c,\ell}}(N_{c,\ell})$ if and only if  $\delta(x)\le \gamma\cdot r_{c,\ell} =\gamma\cdot \rho_j$.
Thus, at that point the quantity $|\{x\in X\mid\delta(x)\le \gamma\cdot \rho_j\}|$ computed by the algorithm is precisely $|X\cap A_{\gamma\cdot r_{c,\ell}}(N_{c,\ell})|$. Thus, whenever the algorithm checks whether this quantity is at least $\ell$, it is checking exactly the condition of $\gamma$-DC-mPJR+ for the pair $(c,\ell)$. 

We now analyze the running time. For a fixed center $c\in M\setminus X$, sorting the $n$ agents by distance from $c$ takes $O(n\log n)$ time. During the scan over the distinct radii, each agent enters the set $P$ exactly once. When an agent enters, the algorithm updates $\delta(x)$ for each center $x\in X$, which takes $O(k)$ time. Hence all updates together take $O(nk)$ time for this fixed unselected center.

When $P$ is updated, computing the coverage takes $O(k)$ time, for a total of $O(nk)$ time. Therefore, the total running time per candidate is $O(n\log n + nk)$. There are at most $m$ centers in $M\setminus X$, so the overall running time is $O(mn\log n + mnk)$.
\end{proof}

Finally, we observe that restricting to default coalitions also resolves the fixed-$\ell$ barrier identified in Section~\ref{sec:limits} (see Proposition~\ref{prop:mpjrplus_fixedell}).
Indeed, we can modify Algorithm~\ref{alg:verify-sc-mpjrplus} to only compute the coverage $\mathrm{cov}$ in the first iteration $j$ such that $t\ge \ell$, and return {\sc False} if $\mathrm{cov} < \ell$
and {\sc True} otherwise.

\begin{corollary}\label{cor:fixed-ell-sc-mpjrplus-check}
Given a center selection $X \subseteq M$ and an integer $\ell \in [k]$, it can be verified in polynomial time whether $X$ satisfies the fixed-$\ell$ version of $\gamma$-DC-mPJR+.
\end{corollary}

This stands in sharp contrast with Proposition~\ref{prop:mpjrplus_fixedell}, which shows that the fixed-$\ell$ version of mPJR+ is coNP-complete. 

To summarize, DC-mPJR+ achieves three properties simultaneously: (1)~it is efficiently verifiable in $O(mn\log n + mnk)$ time, (2)~it implies $(\gamma+2)$-mPJR+, providing an approximate global fairness guarantee, and (3)~it is satisfied by existing algorithms such as SEAR. Together, these features make DC-mPJR+ a practical and theoretically grounded tool for auditing proportional representation in clustering.

\section{Empirical Sanity Check}
\label{sec:empirical-sanity-check}

Since DC-mPJR+ is a relaxation of mPJR+, a natural concern is whether it is too permissive in practice: could it be the case that the default-coalition restriction discards so much of the mPJR+ condition that DC-mPJR+ becomes nearly vacuous? We address this concern through a small empirical study on synthetic two-dimensional Euclidean instances with $M = N$. The goal is not to provide a comprehensive empirical evaluation, but to check whether DC-mPJR+ remains selective on structured instances with natural cohesive groups, and whether its behavior is similar to that of mPJR+ in practice.
 
Carrying out such a comparison requires an mPJR+ verifier that is practical at experimental scale. Since we focus on small values of $k$, we use the $O(mn\log n\cdot 2^k)$ algorithm
from \Cref{prop:direct-mpjrplus-verifier}. 
In our implementation, we replace the possibly fractional quota $q = n/k$ with the equivalent integer comparison $|S(c, r, Y)| \cdot k \ge (|Y| + 1)\, n$.

\subsection{Experimental Setup}
\label{subsec:experiments-setup}

\paragraph{Instance generation.}
For each $n \in \{20, 50, 80, 100\}$ and each number of latent clusters $g \in \{4, 5, 6\}$, we generate $50$ independent Euclidean instances as follows. We first sample $g$ cluster centers uniformly from $[0, 1]^2$, then assign approximately $n / g$ agents to each cluster and perturb each agent's location by isotropic Gaussian noise with standard deviation $0.04$. The candidate set is identified with the agent set, so $M = N$. This clustered model creates natural cohesive groups around the sampled centers, providing meaningful structure against which both axioms can be tested. Since $k = 5$ and all tested values of $n$ are divisible by $5$, the quota is exactly $q = n / k$. For each instance, we sample $1{,}000$ center selections uniformly at random from all size-$5$ subsets of $M$, yielding $50{,}000$ sampled selections per $(n, g)$ pair.

\paragraph{Verification.}
We check DC-mPJR+ using Algorithm \ref{alg:verify-sc-mpjrplus}, and mPJR+ using Algorithm~\ref{alg:verify-mpjr-smallk}.

\subsection{Results}
\label{subsec:experiments-results}

\Cref{fig:sanity-check} reports the fraction of sampled center selections satisfying each axiom across all $(n, g)$ settings. As expected, DC-mPJR+ is consistently more permissive than mPJR+, but the gap remains moderate throughout. Across all configurations, satisfaction rates range from $10.3\%$ to $49.4\%$ for mPJR+ and from $10.5\%$ to $54.4\%$ for DC-mPJR+.

\begin{figure}[h]
    \centering
    \includegraphics[width=1.0\linewidth]{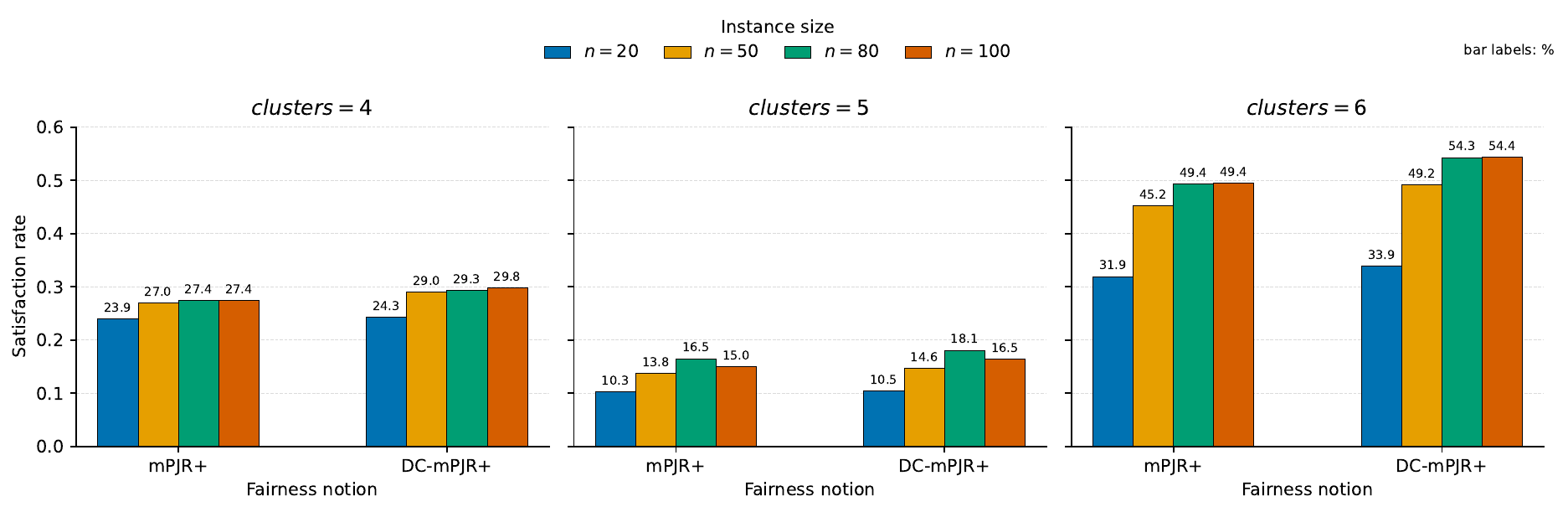}
    \caption{Synthetic sanity check with $M = N$ and $k = 5$. Each subplot fixes the number of latent Gaussian clusters $g \in \{4, 5, 6\}$; bars show the fraction of sampled center selections satisfying mPJR+ and DC-mPJR+ across $n \in \{20, 50, 80, 100\}$.}
    \label{fig:sanity-check}
\end{figure}

Two observations are worth highlighting. First, DC-mPJR+ does not become vacuous on these structured instances: even after restricting to default coalitions, it continues to reject a substantial fraction of uniformly sampled selections. The clustered model is particularly informative here, since $g$ balanced clusters paired with $k = 5$ committee size make it plausible that a random selection fails to allocate representation across the natural cluster structure, and the low satisfaction rates confirm that both axioms are sensitive to such under-representation.

Second, the moderate gap between the two axioms suggests that, on these instances, most mPJR+ violations are already witnessed by default coalitions. This empirically supports the theoretical role of DC-mPJR+ established in \Cref{prop:sc-mpjrplus-implies-mpjrplus}: not only is DC-mPJR+ provably close to mPJR+ in the worst case, it also tracks mPJR+ closely in practice.
\section{Discussion}
\label{sec:discussion-conclusion}

\paragraph{Standard objectives need not ensure proportional representation.}
Classical objectives such as $k$-means and $k$-median optimize aggregate fit rather than proportional representation. \citet{aziz2024proportionally} gave an explicit Euclidean construction (Example~4 in their paper) showing that both $k$-means and $k$-median can produce solutions violating PRF. We use the same construction to show that these methods may violate mPJR+ and DC-mPJR+. Figure~\ref{fig:standard-objectives-fail} illustrates the key geometry: a compact coalition of size $2q$ is assigned only one center by $k$-means/$k$-median, while a smaller but more dispersed group of size $q$ receives two, because the dispersed group contributes more to the aggregate objective. The $2q$ agents in the compact coalition, together with an unselected center nearby, witness a DC-mPJR+ violation (and hence an mPJR+ violation) at level $\ell = 2$, since only one selected center is within the required radius of their coalition.

\begin{figure}[h]
    \centering
    \includegraphics[width=0.7\linewidth]{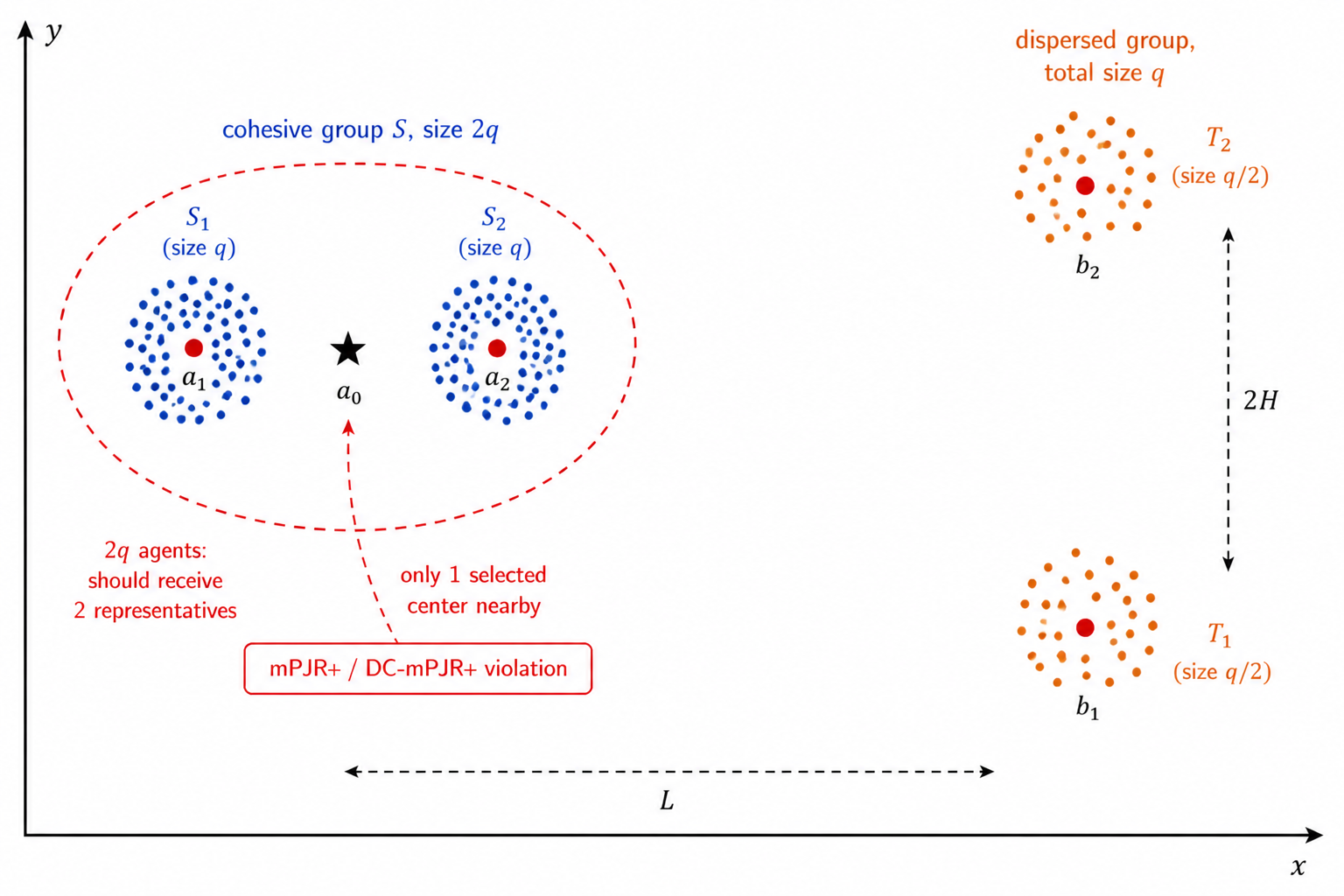}
    \caption{
    Adapted from Example~4 of \citet{aziz2024proportionally}. We assume that $L$ is much larger than $2H$, and $2H$ is much larger than the diameter of $S$. On this instance, 
    \(k\)-means/\(k\)-median selects \(\{a_0,b_1,b_2\}\), violating
    mPJR+ and DC-mPJR+ (witnessed by $c = a_1$, $\ell=2$, and $N_{c, \ell} = S$); proportional center selections include
    \(\{a_1,a_2,b_1\}\) or \(\{a_1,a_2,b_2\}\).
    }
    \label{fig:standard-objectives-fail}
\end{figure}


\paragraph{Limitations and future directions.}
The synthetic experiment is intended as a sanity check rather than a comprehensive empirical evaluation. 
Its purpose is to test whether DC-mPJR+ remains selective, not to establish empirical performance across domains. 
Future work could apply our verification algorithm to real datasets where representation across latent groups is important, as a means of measuring the success of currently deployed clustering algorithms. 

Our work allows proportional representation in clustering to be studied, not just as an objective of algorithm design, but as an auditable property of existing solutions. 
On the theoretical side, a natural question is whether the default coalition approach can be useful for tractable verification of other proportional representation axioms \citep{ebadian2025boosting,kalayci2024proportional}.
For example, one could investigate the approximate notion of $q$-\emph{core}, which requires that no coalition can collectively deviate to alternative centers making every member strictly better off \citep{ebadian2025boosting}. 
Default coalitions may provide a useful starting point, though the core's utility improvement requirement differs structurally from the fixed utility level $\ell$ in PJR, and whether an analogous efficiently verifiable relaxation exists remains open.

\bibliography{citation}
\bibliographystyle{plainnat}


\appendix
\section*{Organization of the Appendix}
The appendix is organized as follows.

\begin{itemize}[leftmargin=10pt]
    \item \textbf{Appendix~\ref{app:embedding}}: describes the approval-to-metric embedding (Lemma~\ref{lem:approval-metric}) and proves coNP-completeness of mPJR verification (Theorem~\ref{thm:mpjr-hard}). These results establish the foundational connection between discrete approval voting and metric fairness that underlies all subsequent hardness results.
    
    \item \textbf{Appendix~\ref{app:mpjrplus}}: show that our definition of mPJR+ follows the Kellerhals--Peters template (Proposition~\ref{prop:kp-template-mpjrplus}), and contains the proofs that SEAR satisfies mPJR+ (Proposition~\ref{prop:sear-pjrplus}) and that mPJR+ is polynomial-time verifiable via submodular minimization (Proposition~\ref{prop:pjrplus-check}).

    \item \textbf{Appendix~\ref{app:discrete_hardness}}: establishes the fixed-$\ell$ coNP-hardness of PJR+ verification in the setting of multiwinner voting with approval ballots (Proposition~\ref{prop:discrete_pjrplus_fixedell}). This result, which may be of independent interest, underlies the metric fixed-$\ell$ hardness result (Proposition~\ref{prop:mpjrplus_fixedell}) via the embedding from Appendix~\ref{app:embedding}.
    

\end{itemize}

\section{The Approval-to-Metric Embedding and mPJR Hardness}
\label{app:embedding}

The results in this section establish the connection between multiwinner voting with approval ballots and proportional clustering. The embedding (Lemma~\ref{lem:approval-metric}) is the key technical tool: it allows us to transfer hardness results from the well-studied approval voting setting into our metric framework. We first establish the properties of the embedding, then use it to prove that mPJR is coNP-hard.

\subsection{Proof of \Cref{lem:approval-metric}}\label{proof of lemma}

\lemApprovalMetric*

\begin{proof}
Let $\mathcal{U}:=N\cup M$. Consider the complete bipartite graph with bipartition $(N,M)$, and assign
to each edge $(i,c)\in N\times M$ the weight
\[
w(i,c)=
\begin{cases}
1, & \text{ if }c\in \mathcal A_i,\\
2, & \text{ if }c\notin \mathcal A_i.
\end{cases}
\]
Let $d$ be the shortest-path distance induced by this weighted graph.
Since the graph is connected and all edge weights are positive, $d$ is a metric on $\mathcal{U}$.
Moreover, $d$ only takes integer values, and, by the triangle inequality, $d(x, y)\le 4$
for all $x, y\in N\cup M$.

We claim that for every agent $i\in N$ and every candidate $c\in M$ it holds that
\[
d(i,c)=
\begin{cases}
1, & \text{ if }c\in \mathcal A_i,\\
2, & \text{ if }c\notin \mathcal A_i.
\end{cases}
\]
Indeed, if $c\in \mathcal A_i$, then the edge $(i,c)$ has weight $1$, so $d(i,c)\le 1$.
Since all edge weights are positive, no shorter path exists, and hence $d(i,c)=1$.

On the other hand, if $c\notin \mathcal A_i$, then the direct edge $(i,c)$ has weight $2$.
Any other path from an agent to a candidate in a bipartite graph must have odd length, so has total length at least three. 
Hence $d(i,c)=2$.

Therefore, for every agent $i\in N$ it holds that
\[
B(i,r)\cap M=
\begin{cases}
\emptyset, & \text{ if }r<1,\\[2mm]
\mathcal A_i, & \text{ if }1\le r<2,\\[2mm]
M, & \text{ if }r\ge 2.
\end{cases}
\]
The second statement follows immediately: for every $S\subseteq N$ and every $1\le r<2$, we have
\[
\bigcap_{i\in S}\left(B(i,r)\cap M\right)=\bigcap_{i\in S} \mathcal A_i,
\qquad
A_r(S)
=
\left(\bigcup_{i\in S} B(i,r)\right)\cap M
=
\bigcup_{i\in S} \mathcal A_i.
\]
\end{proof}

\subsection{Proof of \Cref{thm:mpjr-hard}}\label{proof of thm}

\thmMpjrHard*

With the embedding of \Cref{lem:approval-metric} in hand, the hardness of mPJR verification follows by a direct reduction from the hardness of PJR verification in the approval setting. 
The key observation is that the embedding preserves both cohesiveness (intersection of approval sets) and coverage (union of approval sets) at any radius $y \in [1,2)$, so violations transfer in both directions.

\begin{proof}
We first argue that this problem is in coNP. Indeed, 
given an instance $(N, M, k)$ of centroid clustering in a metric space $({\mathcal U}, d)$ and a center selection $X$, a certificate that $X$ violates mPJR consists of a radius $r\in\{d(i, c)\mid i\in N, c\in M\}$, an integer $\ell\in[k]$, and a coalition $S\subseteq N$.
Given such a certificate, we can check in polynomial time that
\[
    \text{ (1) }\quad|S|\ge \ell\cdot q,\qquad\text{ (2) }\quad
    \left|\bigcap_{i\in S}(B(i,r)\cap M)\right|\ge \ell,\qquad \text{ (3) }\quad
    |X\cap A_r(S)|<\ell.
\]
Note that it is sufficient to only check the radii in $R = \{d(i, c)\mid i\in N, c\in M\}$: indeed, if there is an mPJR violation with a certificate $(r, \ell, S)$, where
$r\not\in R$, there is also an mPJR violation with a certificate $(r', \ell, S)$, where $r'=\max\{q\in R: q\le r\}$; this is because $A_r(T)=A_{r'}(T)$ for every $T\subseteq N$.
It then follows that our certificate is of polynomial size, as we assume that the distance $d$ is polynomial-time computable (and hence every element of $R$ can be encoded by polynomially many bits).

It remains to prove coNP-hardness. \citet{aziz2018complexity} consider the following computational problem, which we will refer to as {\sc PJR Verification}:
given an approval instance
${\mathcal I}=(N,M,(\mathcal A_i)_{i\in N}, k)$ and a selected committee $W\subseteq M$ with $|W|=k$, determine whether $W$ satisfies PJR. 
By Theorem~3 of \citet{aziz2018complexity}, 
{\sc PJR Verification} is coNP-complete.

Let $({\mathcal I}, W)$ be an instance
of {\sc PJR verification}, where
${\mathcal I} = (N,M,(\mathcal A_i)_{i\in N}, k)$ and $W$ is a size-$k$ subset of $M$. Apply \Cref{lem:approval-metric} to construct a metric space $(\mathcal{U},d)$ with $\mathcal{U}=N\cup M$, and consider an instance of centroid clustering ${\mathcal I}'=(N, m, k)$ in  
$(\mathcal{U},d)$, together with a center selection $W\subseteq M$.
We claim that $W$ satisfies PJR in $\mathcal I$ if and only if $W$ satisfies mPJR in ${\mathcal I}'$.

\paragraph{($\Rightarrow$)}
Suppose that $W$ violates mPJR in ${\mathcal I}'$. 
Then there exist a radius $r\ge 0$, an integer $\ell\in[k]$, and a set
$S\subseteq N$ such that
\[
\text{ (1) }\quad|S|\ge \ell\cdot q,\qquad\text{ (2) }\quad
\left|\bigcap_{i\in S}(B(i,r)\cap M)\right|\ge \ell,\qquad \text{ (3) }\quad
|W\cap A_r(S)|<\ell.
\]

Note that it has to be the case that $1\le r<2$.
Indeed, if $r<1$, then by the construction of $d$ in \Cref{lem:approval-metric} it holds that
$B(i,r)\cap M=\emptyset$ for all $i\in N$, so $\bigcap_{i\in S}(B(i,r)\cap M)=\emptyset,$
which is incompatible with (2).

Similarly, if $r\ge 2$, then again by \Cref{lem:approval-metric} it holds that 
$B(i,r)\cap M=M$ for all $i\in N$, and hence $A_r(S)=M$.
Therefore, \(|W\cap A_r(S)|=|W|=k\ge \ell,\) which is incompatible with (3). 

Thus, assume that $1\le r<2$. 
Then \Cref{lem:approval-metric} gives
$B(i,r)\cap M= \mathcal A_i$ for all $i\in N$, so
\[
\left|\bigcap_{i\in S} \mathcal A_i\right|
=
\left|\bigcap_{i\in S}(B(i,r)\cap M)\right|
\ge \ell
\]
and
\[
|W\cap \bigcup_{i\in S} \mathcal A_i|
=
|W\cap A_r(S)|
<\ell.
\]
Together with $|S|\ge \ell\cdot q$, this shows that $W$ violates PJR in the approval instance $\mathcal I$.

\paragraph{($\Leftarrow$)}
Now, suppose  that
$W$ violates PJR in $\mathcal I$. Then there exist an integer $\ell\in[k]$ and a set
$S\subseteq N$ such that
\[
|S|\ge \ell\cdot q,\qquad
\left|\bigcap_{i\in S} \mathcal A_i\right|\ge \ell,\qquad
|W\cap \bigcup_{i\in S} \mathcal A_i|<\ell.
\]
Now consider the metric instance ${\mathcal I}'$ at radius $r=1$.
By Lemma~\ref{lem:approval-metric} we have
\[
B(i,1)\cap M=\mathcal A_i \qquad\text{for all } i\in N.
\]
Hence
\[
\left|\bigcap_{i\in S}(B(i,1)\cap M)\right|
=
\left|\bigcap_{i\in S} \mathcal A_i\right|
\ge \ell
\]
and
\[
|W\cap A_1(S)|
=
|W\cap \bigcup_{i\in S} \mathcal A_i|
<\ell.
\]
Since $|S|\ge \ell\cdot  q$, the set $S$ witnesses a violation of mPJR at radius $1$, i.e., $W$ fails mPJR in ${\mathcal I}'$.

This proves the equivalence. Since {\sc PJR Verification} is coNP-complete, the reduction establishes that mPJR verification is coNP-hard. This completes the proof. 
\end{proof}


\section{Computation and Verification of mPJR+}
\label{app:mpjrplus}

In this section, we collect the auxiliary results regarding mPJR+. 
We first show that our definition of mPJR+ is equivalent to applying the Kellerhals--Peters metric-lifting template to the approval-based PJR+ axiom. 
We then prove the two algorithmic results stated in the main text: mPJR+ is achieved by SEAR (Proposition~\ref{prop:sear-pjrplus}) and admits polynomial-time verification via submodular minimization (Proposition~\ref{prop:pjrplus-check}). 

\subsection{Equivalence with the Kellerhals--Peters Lifting Template}
\label{app:kp-template}

\begin{proposition}\label{prop:kp-template-mpjrplus}
For $\gamma = 1$, Definition~\ref{def:mpjrplus} is equivalent to applying the Kellerhals--Peters metric-lifting template to the approval-based PJR+ axiom.
\end{proposition}

\begin{proof}
Applying the Kellerhals--Peters template to PJR+ requires that, for every $r \ge 0$, the approval instance $(N,M,(A_r(i))_{i \in N},k)$ satisfies PJR+: for every $\ell \in [k]$, every unselected candidate $c \in M \setminus X$, and every coalition $S \subseteq N$ with $|S| \ge \ell\cdot q$ and $c \in \bigcap_{i\in S}A_r(i)$, it holds that
\[
    \left|X \cap A_r(S)\right| \ge \ell .
\]
In metric notation, $c \in \bigcap_{i \in S} A_r(i)$ is equivalent to $\max_{i \in S} d(i,c) \le r$. Thus, the lifted axiom can be written as follows: for every $r \ge 0$, $\ell \in [k]$, $c \in M \setminus X$, and $S \subseteq N$ with $|S| \ge \ell\cdot q$, it holds that
\[
    \max_{i \in S} d(i,c) \le r
    \quad\text{implies that}\quad
    |X \cap A_r(S)| \ge \ell .
\]

$(\Rightarrow)$ Suppose $X$ satisfies Definition~\ref{def:mpjrplus} with $\gamma=1$. Fix any $r,\ell,c,S$ satisfying the premise of the lifted axiom, and set $r(c, S):=\max_{i \in S}d(i,c)\le r$. By Definition~\ref{def:mpjrplus}, $|X\cap A_{r(c, S)}(S)|\ge \ell$. Since $r(c, S)\le r$ implies $A_{r(c, S)}(S)\subseteq A_r(S)$, we obtain
\[
    |X\cap A_r(S)|\ge |X\cap A_{r(c, S)}(S)|\ge \ell .
\]

$(\Leftarrow)$ Conversely, suppose the lifted axiom holds. Fix $\ell \in [k]$, $c \in M \setminus X$, and $S \subseteq N$ with $|S|\ge \ell\cdot q$, and set $r(c, S):=\max_{i \in S}d(i,c)$. Then $c\in \bigcap_{i \in S}A_{r(c, S)}(i)$. Applying the lifted axiom at $r=r(c, S)$ yields
\[
    |X\cap A_r(S)|\ge \ell,
\]
which is exactly Definition~\ref{def:mpjrplus} for $\gamma=1$.
\end{proof}

\subsection{Description of Spatial Expanding Approvals Rule (SEAR)}

    The Spatial Expanding Approval Rule (SEAR) selects $k$ centers by gradually expanding the radius within which agents may support candidate centers. 
    Initially, each agent has weight $1$, and the algorithm orders all agent--candidate distances from smallest to largest. 
    At a given radius, a candidate $c$ is eligible if the total remaining weight of agents within that radius from $c$ is at least the quota $q$.
    Starting from the smallest distance, SEAR increases the radius until some candidate becomes eligible. It then selects an eligible candidate with maximum weighted support,
    adds it to the solution, and removes it from the set of available candidate locations.
    After selecting a center $c^*$ at radius $r$, SEAR reduces the weights of the agents within distance $r$ from $c^*$ in an arbitrary fashion so that their total weight decreases by $q$. 
    The procedure repeats until $k$ centers have been chosen. 
    See Algorithm~\ref{alg:SEAR} for the pseudocode.

    \begin{algorithm}[h]
    \caption{{SEAR (Spatial Expanding Approval Rule)} for Discrete Clustering}
    \label{alg:SEAR}
    \KwIn{An instance of centroid clustering $(N, M, k)$ for a metric space $(\mathcal U, d)$ with $|N|=n$}
    \KwOut{A set of $k$ centers}
    
    $w_i\gets 1$ for each $i\in N$
    
    Order the elements of the set $D=\{d(i,c)\mid i\in N,c\in M\}$ as $d_1\leq d_2 \leq \cdots \leq d_{|D|}$
    
    $j\gets 1$
    
    $W\gets \varnothing$
    
    \While{$|W|<k$}{
        $C^*\gets \left\{c\in M\mid \sum_{i\in B(c, d_j) \cap N} w_i\geq n/k\right\}$
        
        \If{$C^*=\varnothing$}{
            $j\gets j+1$
        }
        \Else{
            Select some candidate $c^*$ from $\arg\max_{c\in C^*} \sum_{i\in B(c, d_j)\cap N} w_i$
            
            $W\gets W\cup \{c^*\}$            
            
            $M\gets M\setminus \{c^*\}$          
            
            $N'\gets B(c^*, d_j)\cap N$        
            
            Modify the weights of agents in $N'$ so that their total weight, i.e., $\sum_{i\in N'} w_i$, decreases by exactly $n/k$
        }
    }
    \Return $W$
    \end{algorithm}
    
\subsection{Proof of Proposition~\ref{prop:sear-pjrplus} (SEAR satisfies mPJR+)}\label{proof of propSEAR}


\propSearPjrplus*
\begin{proof}
    Let $X$ be the set returned by SEAR. Suppose, for contradiction, that $X$ violates mPJR+. Then there exist $c \in M \setminus X$, $\ell \in [k]$, and $S \subseteq N$ such that $|S| \ge \ell\cdot q$ and $|X \cap A_{r(c, S)}(S)| < \ell$, where \(r(c, S) := \max_{i \in S} d(i,c)\).
    
    For readability, let $r=r(c, S)$, and set \(s := |X \cap A_{r}(S)|\). Then $s \le \ell-1$. Since each agent starts with a budget of $1$, the total initial budget of the agents in $S$ is at least $|S| \ge \ell\cdot q$.
    
    Now consider the execution of SEAR up to the moment when the global radius is to be increased beyond~$r$. 
    If an agent $i\in S$ pays for a center $j$ when the current radius is at most $r$, then $d(i,j) \le r$, hence $j \in B(i,r) \cap M \subseteq A_{r}(S)$. Therefore, up to radius $r$, agents in $S$ can spend their budget on centers in $X \cap A_{r}(S)$ only.
    Each opened center costs $q$, and there are only $s$ centers that agents in $S$ can contribute to. Hence the total amount spent by agents in $S$ up to radius $r$ is at most $s\cdot q$. It follows that the remaining total budget of agents in $S$ at radius $r$ is at least $\ell\cdot q - s\cdot q \ge q$. Moreover, by the definition of \(r\), for every agent \(i\in S\) we have \(d(i,c)\le r\). Thus every agent in \(S\) approves the unselected center \(c\) at radius \(r\), and agents in $S$ can collectively afford $c$.
    However, SEAR does not increase the radius while some unselected center is affordable, a contradiction.
    Hence no such triple \((c,\ell,S)\) exists, and \(X\) satisfies mPJR+.
\end{proof}

\subsection{Proof of Proposition~\ref{prop:pjrplus-check}}\label{proof of pjrplus-check}

\propPjrplusCheck*

The proof adapts the submodular minimization technique of \citet{brill2023robust} to the metric setting. For each unselected center $c$ and radius $r$, checking for a violation reduces to minimizing a submodular function. Since the number of relevant $(c, r)$ pairs is at most $mn$, the overall procedure runs in polynomial time.

\begin{proof}
    We adapt a submodular minimization argument of \citet[Proposition~2]{brill2023robust} to the metric setting.

    Recall that a function $g: 2^Z\to\mathbb R$ that maps subsets of a set $Z$ to real numbers is {\em submodular} if for all $X, Y\subseteq Z$ with $X\subseteq Y$ and all $z\in Z\setminus Y$ it holds that $g(X\cup\{z\})-g(X)\ge g(Y\setminus\{z\})-g(Y)$. A canonical example of a submodular function is a coverage function: given a finite ground set $G$ and a collection ${\mathcal S}=\{S_1, \dots, S_t\}$ of subsets of $G$, the coverage function $g^*: 2^{[t]}\to\mathbb Z$ is defined as $g^*(I)=|\cup_{i\in I} S_i|$; it is easy to check that $g^*$ is submodular.
    Note that if $g:2^Z\to\mathbb R$ is a submodular function, then the function 
    $g':2^Z\to\mathbb R$ defined as $g'(X)=g(X)+\alpha\cdot |X|$ for some $\alpha\in\mathbb R$
    is submodular, too; this follows immediately from the definition of submodularity.
    
    Fix a center $c \in M \setminus X$ and a radius $r \ge 0$. Consider the agents who approve $c$ at radius $r$, i.e., the set $B(c,r)\cap N$. For every subset $S \subseteq B(c,r)\cap N$, define
    \[
    f_{c,r}(S):=|X \cap A_r(S)|-\frac{|S|}{q}.
    \]
    We first show that $f_{c,r}$ is submodular. We have $A_r(S)=\left(\bigcup_{i\in S} B(i,r)\right)\cap M = \bigcup_{i\in S} A_r(i)$. Hence, the set function $f': S \mapsto |X\cap A_r(S)|$ is a coverage function: it counts how many members of the fixed set $X$ are covered by the family of sets $\{A_r(i)\}_{i\in S}$. Therefore it is submodular. Since $f_{c,r}(S)=f'(S)+\alpha\cdot |S|$ for $\alpha=\frac{1}{q}$, it follows that $f_{c, r}$ is submodular, too.

    We will now argue that there exists a violation of mPJR+ if and only if $f(c, r)\le -1$ for some $c\in M\setminus X$ and some $r\in \{d(i,c)\mid i\in N\}$.

    Indeed, suppose there is a violation of mPJR+, witnessed by a center $c\in M\setminus X$, an $\ell\in [k]$, 
    and a set $S$ with $|S|\ge \ell\cdot q$.
    Let $r=r(c, S)=\max_{i\in S}d(i, c)$. Note that $r\in \{d(i,c)\mid i\in N\}$.
    We claim that in this case $f_{c, r}(S)\le -1$. Indeed, we have $|X\cap A_r(S)|\le \ell-1$ and hence
    $$
    f_{c, r}(S)= |X\cap A_r(S)| -\frac{|S|}{q} \le \ell-1-\ell = -1.
    $$
    Conversely, suppose $f_{c, r}(S)\le -1$ for some $c\in M\setminus X$, $r\in \{d(i,c)\mid i\in N\}$ and $S\subseteq B(c, r)\cap N$. Let $\ell:=\left\lfloor \frac{|S|}{q}\right\rfloor$; then $|S| \ge \ell\cdot q$. We claim that the triple $(c, \ell, S)$
    is a witness that mPJR+ is violated. Indeed, let $r(c, S)=\max_{i\in S}d(i, c)$. Since $S\subseteq B(c, r)$, we have $r(c, S)\le r$. From $f_{c, r}(S)\le -1$ we obtain
    $$
    |X\cap A_r(S)| -\frac{|S|}{q}\le -1;
    $$
    as $\frac{|S|}{q}\ge \ell$, this implies $|X\cap A_r(S)|\le \ell-1$. As we have established that $r(c, S)\le r$, it follows that 
    $|X\cap A_{r(c, S)}(S)|\le \ell-1$, i.e., there is a violation of mPJR+.
 
    
    
    
    
We conclude that we can verify mPJR+ by iterating over all unselected centers $c\in M\setminus X$ and all $r\in \{d(i,c)\mid i\in N\}$, minimizing the corresponding submodular function $f_{c,r}$ on $2^{B(c, R)\cap N}$, and checking if the resulting minimum value is at most $-1$. Since there are at most $mn$ functions to consider, 
and submodular minimization admits a polynomial-time algorithm~\citep{lee2015faster}, 
the verification terminates in polynomial time.
\end{proof}

\section{Fixed-$\ell$ Hardness of PJR+ and mPJR+}
\label{app:discrete_hardness}

In Section~\ref{sec:limits}, we stated that verifying mPJR+ for a fixed representation level $\ell$ is intractable (Proposition~\ref{prop:mpjrplus_fixedell}). 
The proof of that result relies on a corresponding hardness result in the setting of multiwinner voting with approval ballots, which we state and prove here.
Beyond its role in our reduction, this result may be of independent interest to the social choice community, as it rules out a very natural approach to devising more efficient algorithms for verification of PJR+.
\subsection{Hardness Proof for Fixed-$\ell$ PJR+}
We refer the reader to Section~\ref{sec:preliminaries} for background on multiwinner voting with approval ballots.

\begin{proposition}\label{prop:discrete_pjrplus_fixedell}
    Given an instance $(N, M, ({\mathcal A}_i)_{i\in N}, k)$ of multiwinner voting with approval ballots, a selected committee $X$, and a fixed integer $\ell\in[k]$, it is coNP-complete to decide whether $X$ satisfies the fixed-$\ell$ PJR+ condition; that is, whether for every
    unselected candidate $c\notin X$ and every coalition of voters $S\subseteq N$ such that
    \[
        |S|\ge \ell\cdot q
        \quad\text{and}\quad
        c\in \bigcap_{i\in S}\mathcal A_i
    \]
    it holds that
    \[
        \left|X\cap \bigcup_{i\in S}\mathcal A_i\right|\ge \ell .
    \]
\end{proposition}

\begin{proof}
    We first observe that this problem is in coNP.
    Indeed, a certificate that the fixed-$\ell$ PJR+ condition is violated consists of an unselected candidate
    $c\notin X$ and a coalition of voters $S\subseteq N$.
    Given such a certificate, we can check in polynomial time that
    \[
        \text{ (1) }\quad|S|\ge \ell\cdot q,\qquad\text{ (2) }\quad
        c\in \bigcap_{i\in S}\mathcal A_i,
        \qquad \text{ (3) }\quad
        \left|X\cap \bigcup_{i\in S}\mathcal A_i\right|<\ell .
    \]

    To prove coNP-hardness, it suffices to show that the complementary problem of 
    deciding whether there exists a violating pair $(c,S)$ for the fixed-$\ell$ PJR+ condition is NP-hard. 
    
    We reduce from the \textsc{Balanced Biclique} problem. An instance of this problem is given by a bipartite graph $G=(L,R,E)$ and an integer $t$. It is a yes-instance if there exist sets $L'\subseteq L$ and $R'\subseteq R$ with $|L'|=|R'|=t$ such that every vertex in $L'$ is adjacent to every vertex in $R'$, and a no-instance otherwise.

    We first show that we may assume, without loss of generality, that the input biclique instance satisfies $|L|=|R|=2t-1$. 
    Starting from an arbitrary instance $(G=(L,R,E),t)$ of \textsc{Balanced Biclique},
    let
    \[
    n' := \max\{|L|,|R|,2t-1\}.
    \]
    We first add isolated vertices to the smaller side, and if necessary to both
    sides, so that both sides have size exactly $n'$. This does not change whether
    the graph contains a $t\times t$ biclique.
    
    Next, let
    \[
    p := n' - 2t + 1 \ge 0
    \qquad\text{and}\qquad
    t' := t+p = n'-t+1.
    \]
    We add $p$ new vertices to each side, and connect each new vertex to all vertices on the other side; we will refer to these added vertices as {\em universal} vertices.
    Let the resulting graph be
    $G'=(L',R',E')$. Then
    \[
    |L'|=|R'|=n'+p = 2(n'-t+1)-1 = 2t'-1.
    \]
    
    We claim that $G$ contains a $t\times t$ biclique if and only if $G'$ contains
    a $t'\times t'$ biclique. 
    Indeed, if $G$ contains a $t\times t$ biclique, then adding all $p$ universal
    vertices on both sides gives a $(t+p)\times(t+p)=t'\times t'$ biclique in $G'$.
    Conversely, suppose that \(G'\) contains a \(t'\times t'\) biclique. Remove all universal
    vertices from this biclique; the remaining biclique then contains at least $t'-p$ 
    vertices on each side. Moreover, none of these vertices can be among the isolated vertices added during the first stage. Hence, $G$ contains a biclique with at least $t'-p=t$ vertices on each side.  
    
    Therefore, the restricted version of \textsc{Balanced Biclique} in which
    $|L|=|R|=2t-1$ remains NP-hard. In the rest of the proof, we rename $G'$ as $G$
    and $t'$ as $t$, and assume that $|L|=|R|=2t-1$.

    Now, fix such an instance $G=(L,R,E)$ with $|L|=|R|=2t-1$. 
    We construct an instance of multiwinner voting with approval ballots $(N, M, (\mathcal A_i)_{i\in N}, k)$ as follows:
    
    \begin{itemize}[leftmargin=10pt]
        \item \textbf{Voters:} For each $u \in L$, we create a voter $v_u$. Thus, the set of voters is $N := \{v_u \mid u \in L\}$. The total number of voters is $n = 2t-1$.
        \item \textbf{Candidates:} For each $w \in R$, we create a candidate $c_w$. We also create one additional candidate $z$. Thus, the set of candidates is $M := \{c_w \mid w \in R\} \cup \{z\}$.
        \item \textbf{Committee:} Let the selected committee be $X := M \setminus \{z\} = \{c_w \mid w \in R\}$. The committee size is $k = |X| = 2t-1$.
        \item \textbf{Representation Level:} We set the fixed representation level to be $\ell := t$. 
    \end{itemize}
    
    Since $n = 2t-1$ and $k = 2t-1$, the quota required for one representative is $q = \lceil n/k \rceil = 1$. The required coalition size for level $\ell$ is therefore $\ell \cdot q = t$.
    
    \begin{itemize}[leftmargin=10pt]
        \item \textbf{Approval Sets:} Each voter $v_u \in N$ approves the unselected candidate $z$, and $v_u$ approves $c_w$ if and only if there is \emph{no} edge between $u$ and $w$ in $G$. Formally,
        \[
        \mathcal A_{v_u} := \{z\} \cup \{c_w\mid \{u,w\} \notin E\}.
        \]
    \end{itemize}

    We claim that $G$ contains a $t\times t$ biclique if and only if there exists an unselected candidate $c \notin X$ and a coalition of voters $S \subseteq N$ that witness a violation of the fixed-$\ell$ PJR+ condition; that is, $|S| \ge \ell\cdot q = t$, every voter in $S$ approves $c$, and $|X \cap \bigcup_{i \in S} \mathcal A_i| < \ell = t$.

    Since $z$ is the only unselected candidate, any violating coalition must agree on $c = z$. Note that by our construction, all voters approve $z$, so the agreement condition is trivially satisfied for any coalition.

    \medskip
    \noindent
    \textbf{($\Rightarrow$) Graph has a biclique $\implies$ Committee violates fixed-$\ell$ PJR+:} 
    
    Suppose that $G$ contains a $t\times t$ biclique, i.e., there exist $L'\subseteq L$ and $R'\subseteq R$ with $|L'|=|R'|=t$ such that every vertex in $L'$ is adjacent to every vertex in $R'$.
    
    Consider the coalition $S := \{v_u \mid u \in L'\}$ and the candidate set $Y=\{c_w\mid w\in R'\}$. Clearly, $|S| = t = \ell\cdot q$ and $|Y|=t$. By construction, no voter in $S$ approves any of the candidates in $Y$.
    Therefore, the voters in $S$ collectively approve at most 
    $$
    |X\setminus Y| = |X|-|Y|=2t-1-t=t-1<\ell
    $$ 
    candidates in $X$.
    Thus, $z$ and $S$ witness a violation of the fixed-$\ell$ PJR+ condition.

    \medskip
    \noindent
    \textbf{($\Leftarrow$) Committee violates fixed-$\ell$ PJR+ $\implies$ Graph has a biclique:}
    
    Suppose there exists a coalition $S \subseteq N$ such that $|S| \ge \ell\cdot q = t$ and $\left| X \cap \bigcup_{i \in S} \mathcal A_i \right| < \ell = t$.
    Let $L'=\{u\in L\mid v_u\in S\}$, 
    $R'=\{w\in R\mid c_w\in X\setminus \bigcup_{i \in S} \mathcal A_i\}$.
    We claim that $(L', R')$ forms a biclique in $G$.
    Indeed, for each $u\in L'$ and $w\in R'$ it holds that $c_w$ is not 
    approved by $v_u$ and hence $\{u, w\}\in E$.
    Moreover, we have $|L'|=|S|\ge t$, $|R'|> |X|-t=t-1$ and hence $|R'|\ge t$. Thus, we can pick $t$ vertices from $L'$ and $t$ vertices from $R'$ to obtain a $t$-by-$t$ biclique in $G$. 

    This completes the reduction. Since \textsc{Balanced Biclique} is NP-hard, deciding whether a fixed-$\ell$ PJR+ violation exists is NP-hard. Therefore, deciding whether no such violation exists, equivalently whether $X$ satisfies the fixed-$\ell$ PJR+ condition, is coNP-hard.
\end{proof}

\subsection{Transfer to the Metric Setting (Proof of Proposition~\ref{prop:mpjrplus_fixedell})}\label{proof of mpjrplus-fixell}

\propMpjrplusFixedell*

\begin{proof}
We first show that the problem is in coNP. 
A certificate that \(X\) violates the fixed-\(\ell\) mPJR+ condition consists of an unselected center \(c\in M\setminus X\) and a coalition of agents \(S\subseteq N\). 
Given such a certificate, we compute
\[
    r(c, S) := \max_{i\in S} d(i,c).
\]
Note that $r(c, S)$ is polynomial-time computable since we assumed that the distance $d$ is efficiently computable.
We can then check in polynomial time that
\[
    |S|\ge \ell\cdot q
    \qquad\text{and}\qquad
    |X\cap A_{r(c, S)}(S)|<\ell .
\]
Indeed, the first condition is straightforward to verify, and the second can be checked by scanning all centers in \(X\) and, for each $x\in X$, checking whether \(x\in A_{r(c, S)}(S)\), i.e., whether there exists an agent \(i\in S\) with \(d(i,x)\le r(c, S)\). 
Thus, non-satisfaction of the fixed-\(\ell\) mPJR+ condition has a polynomial-size certificate that can be verified in polynomial time. 
Therefore, fixed-\(\ell\) mPJR+ verification is in coNP.

It remains to prove coNP-hardness. 
To this end, 
we reduce from fixed-$\ell$ PJR+ verification problem, which was shown to be coNP-hard in Proposition~\ref{prop:discrete_pjrplus_fixedell}.

Let $(N,M,(\mathcal A_i)_{i\in N},k)$ together with $\ell\in[k]$ and $X\subseteq M$ be an arbitrary instance of the decision problem from Proposition~\ref{prop:discrete_pjrplus_fixedell},
where the question is whether there exists an unselected candidate $c\in M\setminus X$ and a coalition
$S\subseteq N$ such that
\(|S|\ge \ell\cdot q\), \(c\in \bigcap_{i\in S} \mathcal A_i\), and \( \left|X\cap \bigcup_{i\in S} \mathcal A_i\right|<\ell \).
We use Lemma~\ref{lem:approval-metric} to construct a metric space $(\mathcal{U},d)$ with $\mathcal{U}=N\cup M$, where for each $i\in N$, $c\in M$ we have $d(i, c)=1$ if $c\in{\mathcal A}_i$ and $d(i, c)=2$ otherwise.
We keep the same agent set $N$, candidate set $M$, target size $k$, center selection $X$ and level $\ell$.

We claim that the approval instance has a fixed-\(\ell\) PJR+ violation if and
only if the constructed metric instance has a fixed-\(\ell\) mPJR+ violation.

\paragraph{($\Rightarrow$)}
Suppose there exist $c\in M\setminus X$ and $S\subseteq N$ such that
\[
|S|\ge \ell\cdot q,\qquad
c\in \bigcap_{i\in S} \mathcal A_i,\qquad
\left|X\cap \bigcup_{i\in S} \mathcal A_i\right|<\ell.
\]
Since every agent in $S$ approves $c$, Lemma~\ref{lem:approval-metric} implies that
$d(i,c)=1$ for all $i\in S$. Hence $r=\max_{i\in S} d(i,c)=1$. 
Again by Lemma~\ref{lem:approval-metric}, $A_r(S)= A_1(S)=\bigcup_{i\in S} \mathcal A_i$. Therefore
\[
|X\cap A_r(S)|
=
\left|X\cap \bigcup_{i\in S} \mathcal A_i\right|
<\ell.
\]
Thus the same pair $(c,S)$ witnesses a violation in the metric instance.

\paragraph{($\Leftarrow$)} Suppose there exist \(c \in M \setminus X\) and \(S \subseteq N\) such that \(|S| \ge \ell\cdot q\) and $|X \cap A_r(S)| < \ell$, where $r = \max_{i\in S} d(i,c)$. Note that, by  construction of $(\mathcal U, d)$, we have $r\in\{1, 2\}$.
Further, if \(r=2\), then $A_r(S)=A_2(S)=M$, and hence $|X\cap A_r(S)|=|X|=k\ge \ell$, which contradicts \(|X\cap A_r(S)|<\ell\). Thus,  \(r=1\).

Since \(r=1\), we have \(d(i,c)=1\) for every \(i\in S\). By construction,
this means that \(c\in \mathcal A_i\) for every \(i\in S\), or equivalently
\[
    c\in \bigcap_{i\in S} \mathcal A_i.
\]
Moreover, at radius \(1\), Lemma~\ref{lem:approval-metric} gives
\[
    A_r(S)=A_1(S)=\bigcup_{i\in S} \mathcal A_i.
\]
Hence
\[
    \left|X\cap \bigcup_{i\in S} \mathcal A_i\right|
    =
    |X\cap A_r(S)|
    < \ell.
\]
Together with \(|S|\ge \ell\cdot q\), this shows that the same pair \((c,S)\)
witnesses a fixed-\(\ell\) PJR+ violation in the approval instance.

We have argued that a no-instance of the approval-based verification problem maps to a no-instance of the metric verification problem and vice versa. 
Therefore, deciding
whether \(X\) satisfies fixed-\(\ell\)
mPJR+ is coNP-hard.
\end{proof}

\end{document}